\documentclass[a4paper,UKenglish]{article}
\usepackage[english]{babel}
\usepackage[utf8]{inputenc}
\usepackage[T1]{fontenc}
\usepackage{lmodern}

\usepackage{amsmath,amsfonts,amsthm,amssymb}

\usepackage{subcaption}
\usepackage{graphicx}
\graphicspath{{./fig/}}

\usepackage{hyperref}
\usepackage[backend=biber, style=alphabetic, sorting=nymt, maxbibnames=5]{biblatex} %% for bibliography
\DeclareSortingTemplate{nymt}{ % custom sorting for biblio
  \sort{
    \field{presort}
  }
  \sort[final]{
    \field{sortkey}
  }
  \sort{
    \field{sortname}
    \field{author}
  }
  \sort{
    \field{sortyear}
    \field{year}
  }
  \sort{
    \field[padside=left,padwidth=2,padchar=0]{month}
    \literal{00}
  }
  \sort{
    \field{sorttitle}
 }
}
\usepackage[ruled,vlined,norelsize]{algorithm2e}

\newtheorem{theorem}{Theorem}
\newtheorem{definition}{Definition}
\newtheorem{proposition}{Proposition}
\newtheorem{lemma}{Lemma}

\newtheorem{remark}{Remark}

\renewcommand{\epsilon}{\varepsilon}
\newcommand{\tiling}{\mathcal{T}}
\renewcommand{\phi}{\varphi}
\newcommand{\atlas}{\mathbf{A}}
\graphicspath{{./img/}}%helpful if your graphic files are in another directory

\title{Geometrical Penrose Tilings are characterized by their $1$-atlas} %TODO Please add

\date{2023}

\author{Thomas Fernique \footnote{HSE University, Moscow, Russia} \and Victor Lutfalla \footnote{Université Publique, France \& GREYC, Université de Caen, Caen, France \& LIS, Aix-Marseille Université, Marseille, France}}
\begin{document}

\maketitle

\paragraph{Abstract.}
Penrose rhombus tilings are tilings of the plane by two decorated rhombi such that the decoration match at the junction between two tiles (like in a jigsaw puzzle). 
In dynamical terms, they form a tiling space of finite type. 
If we remove the decorations, we get, by definition, a sofic tiling space that we here call geometrical Penrose tilings.
Here, we show how to compute the patterns of a given size which appear in these tilings by three different methods: two based on the substitutive structure of the Penrose tilings and the last on their definition by the cut and projection method. 
We use this to prove that the geometrical Penrose tilings are characterized by a small set of patterns called vertex-atlas, \emph{i.e.}, they form a tiling space of finite type. 
Though considered as folk, no complete proof of this result has been published, to our knowledge.


\paragraph{Funding.} ANR C\_SyDySi \& RIN DynNet

\section{Introduction}
Penrose tilings are defined as the tilings of the Euclidean plane by the fat and the thin rhombus with arrowed edges of Figure \ref{fig:rhombus_arrows} \cite{penrose1974, penrose1979, grunbaum1987, senechal1996}.
Two rhombi intersect in either a single vertex or a whole edge.
In the latter case, the common edge of the two rhombi must be identically arrowed in each rhombus (orientation and single/double arrow).

\begin{figure}[htp]
  \centering
  \begin{subfigure}[b]{0.3\textwidth}
    \centering
    \includegraphics[width=0.8\textwidth]{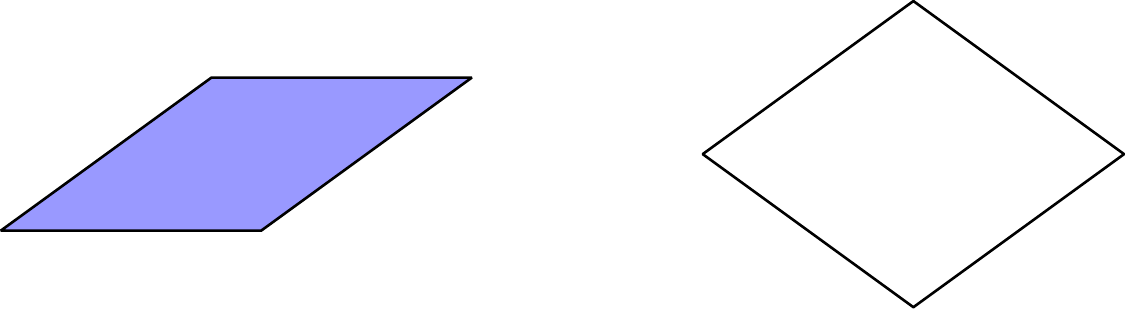}
    \caption{Geometrical tiles.}
    \label{fig:rhombus_geometrical}
  \end{subfigure}
  \begin{subfigure}[b]{0.3\textwidth}
    \centering
    \includegraphics[width=0.8\textwidth]{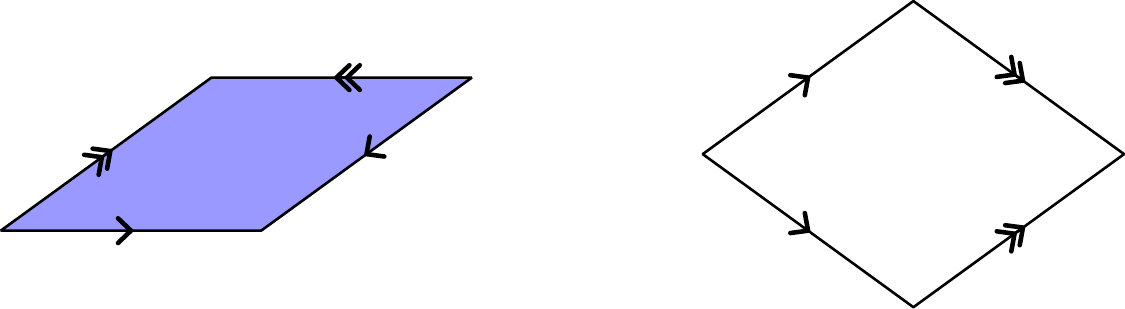}	
    \caption{With arrowed edges.}
    \label{fig:rhombus_arrows}
  \end{subfigure}
  \begin{subfigure}[b]{0.3\textwidth}
    \centering
    \includegraphics[width=0.8\textwidth]{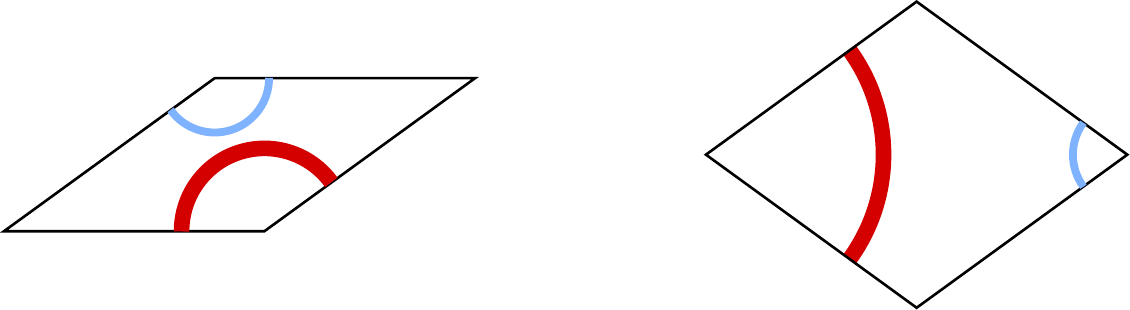}
    \caption{With coloured arcs.}
    \label{fig:rhombus_arcs}
  \end{subfigure}
  \caption{The Penrose rhombus tiles (up to isometry): all edges are of length $1$, the thin rhombus has angles $\tfrac{\pi}{5}$ and $\tfrac{4\pi}{5}$, the fat rhombus has angles $\tfrac{2\pi}{5}$ and $\tfrac{3\pi}{5}$.
  In total there are 20 decorated tiles up to translation, and 10 undecorated tiles up to translation.}
  \label{fig:rhombus_all}
\end{figure}

For better readability of the figures, we consider an alternate definition where arrows are replaced by coloured arcs (Fig.~\ref{fig:rhombus_arcs}) that must match on the edge of adjacent tiles, as in Figure \ref{fig:0_atlas_labels}. The type of the arrow is encoded by the colour of the arc, and the orientation by the intersection point of the arc with the edge which is offset in one direction. %% and not in the middle of the edge.

\begin{figure}[htp] 
  \includegraphics[width=\textwidth]{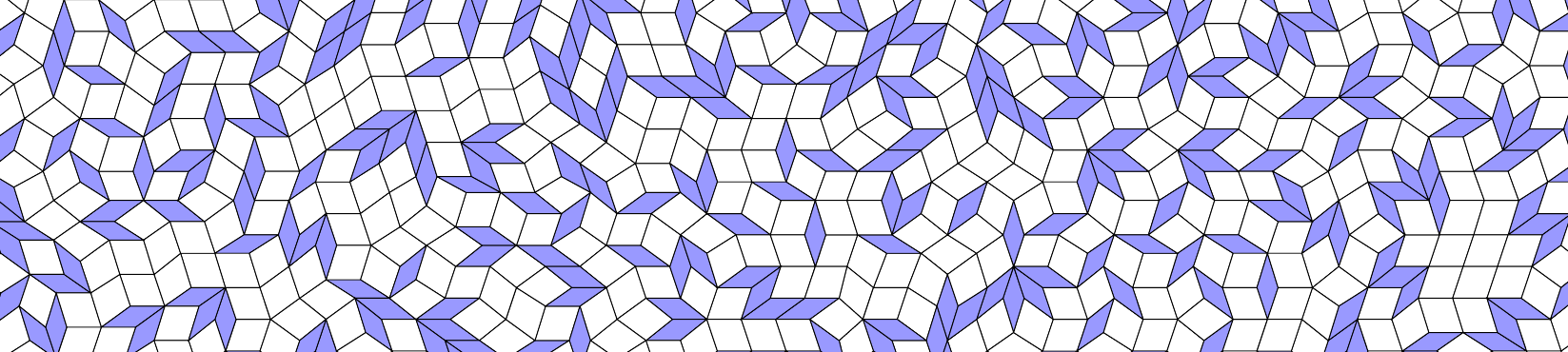}
  \caption{A random tiling with the thin and fat rhombus.}
  \label{fig:random_tiling}
\end{figure}

Penrose tilings are defined with labeled tiles (arrows or coloured arcs) but we often represent them without the labels (Fig.~\ref{fig:penrose}). 
We call \emph{geometrical Penrose tilings} the tilings obtained by removing the labels from Penrose tilings.
Let us emphasize that, of course, most tilings by copies of the thin and fat rhombi are not geometrical Penrose tilings; see for example Fig.~\ref{fig:random_tiling} and \ref{fig:0_atlas_not_sufficient}.
The question we consider is : are the geometrical Penrose tilings characterized by their patterns of a given finite size? 
In other words, we are interested in the quantity of information contained in the labels (coloured arcs or arrows) of Penrose tilings : can these labels be determined locally?

\begin{figure}[htp]
  \includegraphics[width=\textwidth]{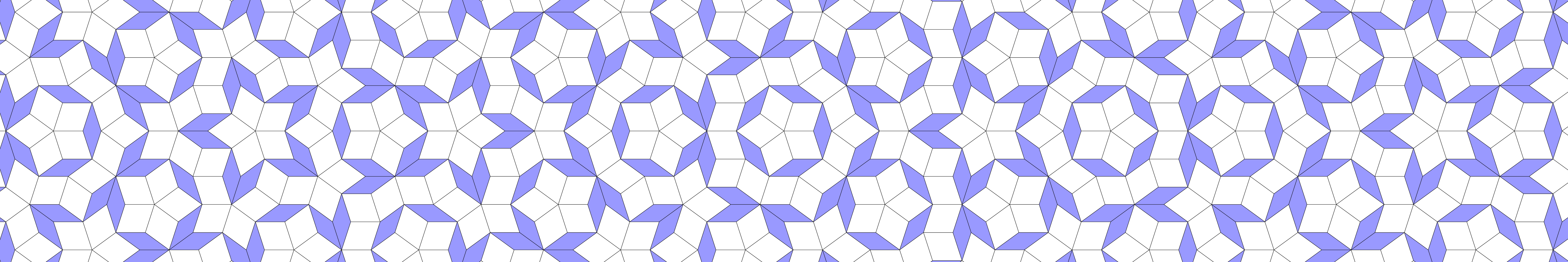}
  \caption{A geometrical Penrose tiling, \emph{i.e.}, a Penrose tiling where arrows or coloured-arcs have been removed.
    Thin rhombi are coloured in blue to improve the readability.}
  \label{fig:penrose}
\end{figure}

Formally, a \emph{patch} is a simply-connected finite set of non-overlapping tiles; and a \emph{pattern} is a patch up to translation.
\begin{figure}[htp]
  \centering
  \includegraphics[width=0.5\textwidth]{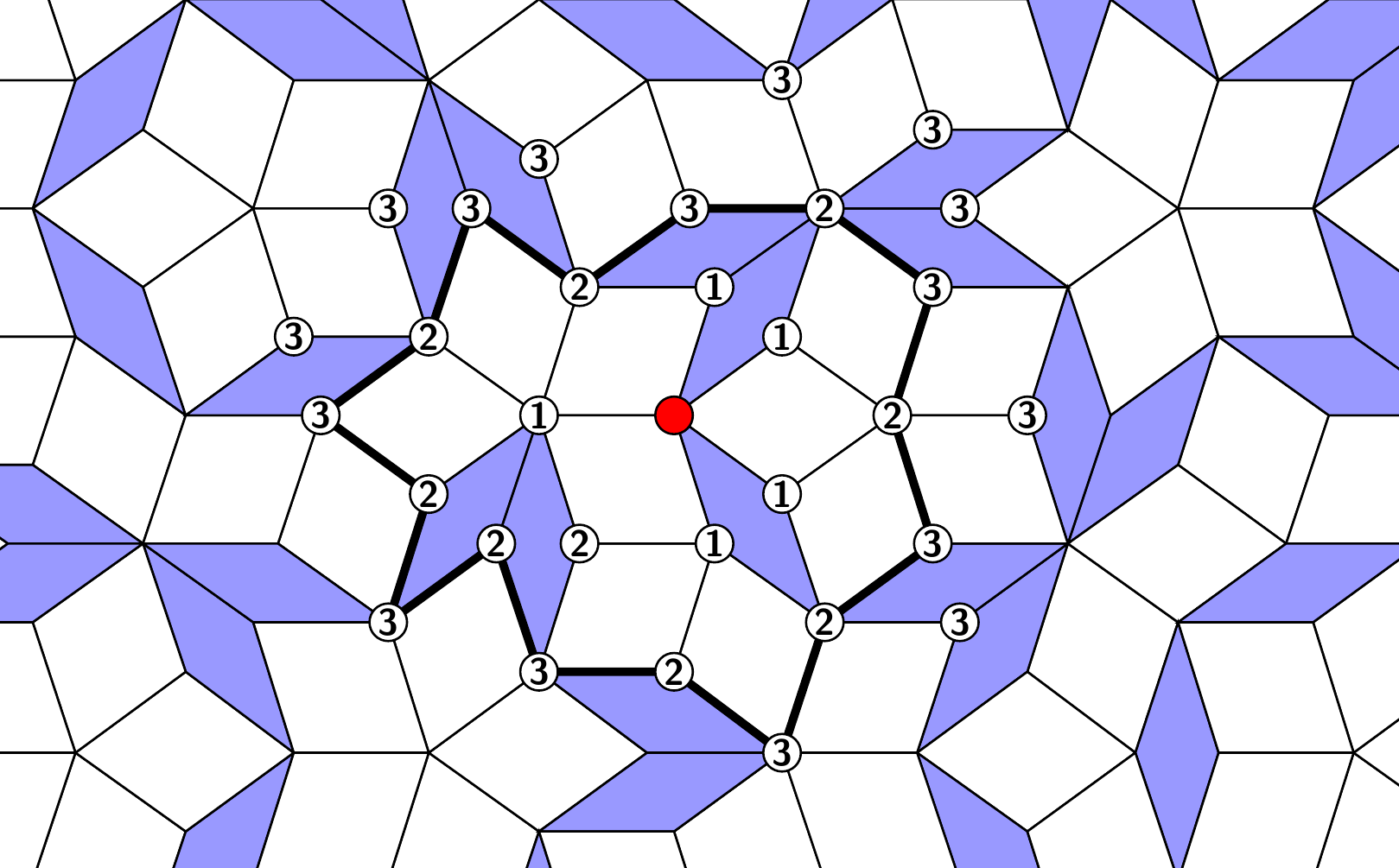}
  \caption{
    A $1$-map centered on the red point, with its boundary emphasized.
    The numbers give the distance to the red point.
  }
  \label{fig:kmap_vertices}
\end{figure}

\begin{definition}[$k$-map and $k$-atlas]
  A {\em $k$-map} of a rhombus tiling is a pattern formed by all the tiles which have a vertex at edge-distance at most $k$ from a given vertex. %% question: pourquoi le vocabulaire k-map
  For example, the rhombi which share a vertex is a $0$-map.
  The set of all the $k$-maps of a tiling (or tiling space) is called the {\em $k$-atlas} or {\em $k$-vertex-atlas}.
  
  A tiling space $X$ is said to be \emph{characterized} by its $k$-atlas $\atlas_k$ when any tiling whose $k$-maps all belong to $\atlas_k$ belongs to $X$.
\end{definition}

\begin{figure}[htp]
  \includegraphics[width=\textwidth]{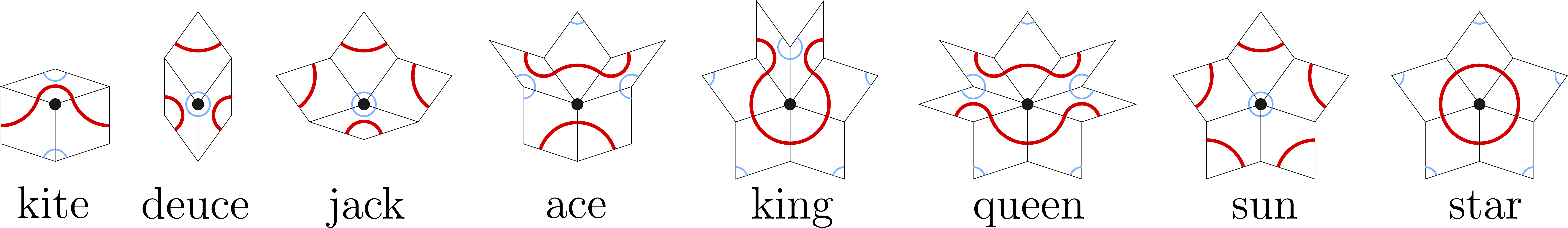}
  \caption{The $0$-atlas of Penrose tilings with coloured arcs (up to isometry), the names of the patterns come from \cite{grunbaum1987}.}
  \label{fig:0_atlas_labels}
\end{figure}

\begin{figure}[htp]
  \center
  \includegraphics[width=0.9\textwidth]{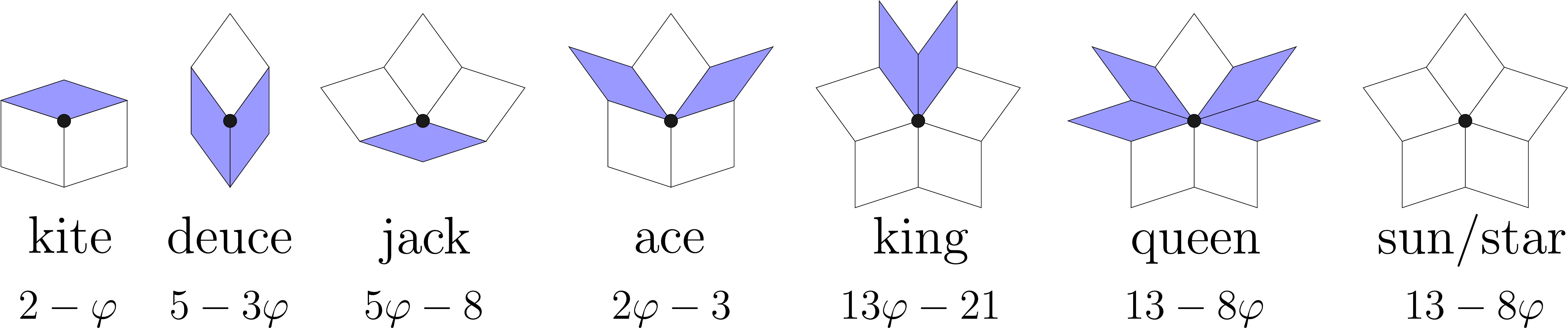}
  \caption{The $0$-atlas of geometrical Penrose tilings (up to isometry), the numbers are the frequencies of the patterns with $\phi=(1+\sqrt{5})/2$ the golden ratio.}
  \label{fig:0_atlas}
\end{figure}

The $0$-atlas of Penrose tilings is known \cite{senechal1996} to contain exactly $8$ $0$-maps with labels (Fig.~\ref{fig:0_atlas_labels}) and $7$ $0$-maps without labels (Fig.~\ref{fig:0_atlas}).
The names of the patterns come from the standardized names of kite-and-dart patterns \cite[Fig. 10.5.3]{grunbaum1987} which are then translated to rhombus patterns by the ``zoom-in'' transformation \cite[Fig. 10.3.19]{grunbaum1987} (note that one can choose instead to use the ``zoom-out'' transformation \cite[Fig. 10.3.14]{grunbaum1987} which yealds a different set of names as in \cite[Fig. 7]{debruijn1981}).

Despite a common belief, the $0$-atlas does not characterize geometrical Penrose tilings, that is, there exist tilings of the whole plane whose $0$-maps all belongs to the $0$-atlas of Penrose tilings but which are not Penrose tilings (Fig.~\ref{fig:0_atlas_not_sufficient}).

In \cite{senechal1996} (Theorem 6.1 p.~177) it is stated that, to characterize Penrose tilings, it suffices to add the condition that two neighbour $0$-maps that share a rhombus cannot be related by a half-turn around the center of that rhombus. 
However the proof is incomplete as only the case around a sun/star pattern (see Fig. \ref{fig:0_atlas}) is proved though it is not the only case to consider (as proved by the example Fig. \ref{fig:0_atlas_not_sufficient} that does not contain the sun/star pattern).
Moreover, the statement and proof use a terminology that conflicts with the standard terminology on Penrose tilings as defined in \cite{penrose1974} and \cite{grunbaum1987} which has led to many misunderstandings of the statement and proof. 

\begin{figure}[htp]
  \center
  \includegraphics[width=0.8\textwidth]{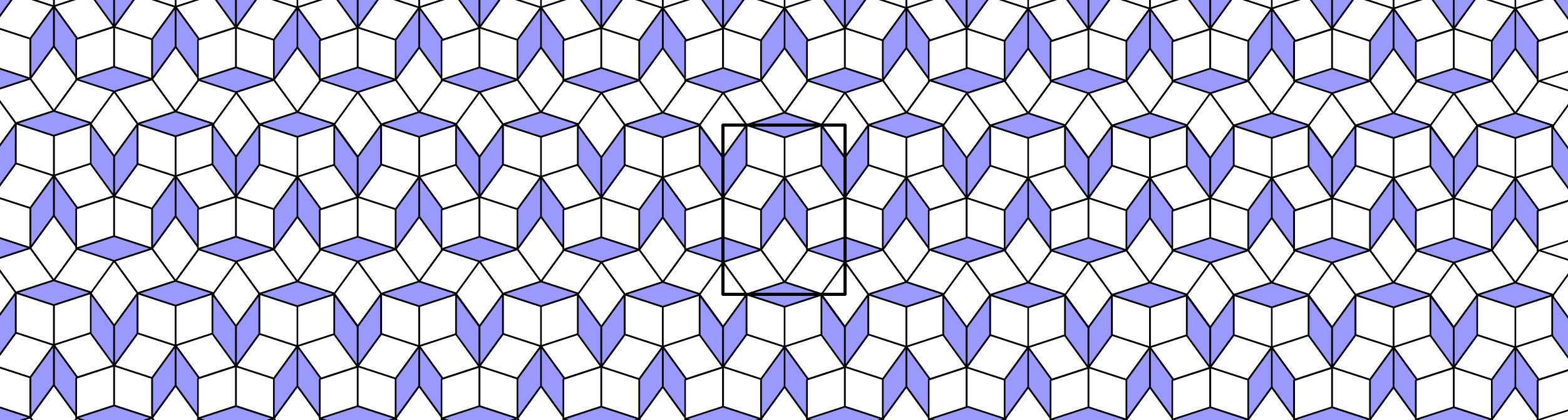}
  \caption{A periodic tiling whose $0$-maps belong to the $0$-atlas of Penrose tilings.}
  \label{fig:0_atlas_not_sufficient}
\end{figure}

Rather than to complete the proof of the statement by Sénéchal, we reformulate it as Th.~\ref{th:main}, which can be considered as \emph{folk.} but we hope is less likely to be misunderstood.

We denote by $\atlas_1$ the $1$-atlas of Penrose tilings presented (up to isometry) in Figure \ref{fig:1_atlas}. 
\begin{theorem}
  \label{th:main}
  Geometrical Penrose tilings are characterized by their $1$-atlas, that is, any tiling by the thin and fat rhombus whose $1$-maps all belong to $\atlas_1$ is a geometrical Penrose tiling.
\end{theorem}

This result can be reformulated in terms of dynamical systems.
We call \emph{tiling space} or \emph{subshift} a set of tilings that is invariant under translation and closed for the tiling topology \cite{robinson2004}.
A \emph{subshift of finite type} (SFT) is a subshift that is characterized by a finite number of forbidden patterns, or equivalently by a vertex atlas.
The tiling space of Penrose tilings contains uncountably many different tilings \cite{dolbilin1995}, none of which is periodic.
They all have the same finite patterns, that is, they are locally indistinguishable \cite{grunbaum1987}.
In dynamical terms, this theorem can be reformulated as :  the tiling space of geometrical Penrose tilings (with removed labels) has {\em finite type}.

In Section \ref{sec:finding} we present the 1-atlas of Penrose tilings and prove that the 1-atlas we present in Figure \ref{fig:1_atlas} is exact.
To achieve this we present three proofs, two based on the substitutive definition of Penrose tilings \cite{penrose1974} and the other based on the cut-and-project definition of Penrose tilings \cite{debruijn1981}, the second proof provides frequencies for the $0$-maps and $1$-maps.
In Section \ref{sec:characterizes} we prove that this 1-atlas indeed characterizes Penrose tilings.

Though the main result can be considered as known or folklore, the three proofs we present are new and highlight three different aspects of Penrose tilings. In particular we provide in these proofs two new results : the frequencies of appearance for the 0-maps and 1-maps, and an explicit bound for the ratio of linear recurrence for Penrose tilings. 

The techniques presented in Section \ref{sec:finding} can be adapted to compute the globally allowed patterns of a given size from any tiling that is defined by a primitive substitution (Section \ref{subsec:graph_substitution}), by a primitive substitution with a similarity expansion (Section \ref{subsec:substitution}) or by cut-and-projection (Section \ref{subsec:cut-and-project}).
Note however that, there exist substitution cut-and-projecttiling that are not characterized by their 1-atlas, for example the Amman-Beenker1982 tiling \cite{beenker1982} is not characterized by its unlabelled patterns of any given finite size \cite{bedaride2015b}, so Section \ref{sec:characterizes} does not generalize to any substitution cut-and-project tiling. 

\section{Finding the 1-atlas of Penrose tilings}
\label{sec:finding}
The first step to prove Theorem~\ref{th:main} is to find the $1$-atlas of Penrose tilings.
\begin{proposition}
  The $1$-atlas of geometrical Penrose tilings $\atlas_1$ is exactly the 15 $1$-maps up to isometry represented in Figure \ref{fig:1_atlas}.
  \label{prop:1_atlas}
\end{proposition}

\begin{figure}[!htp]
  \includegraphics[width=\textwidth]{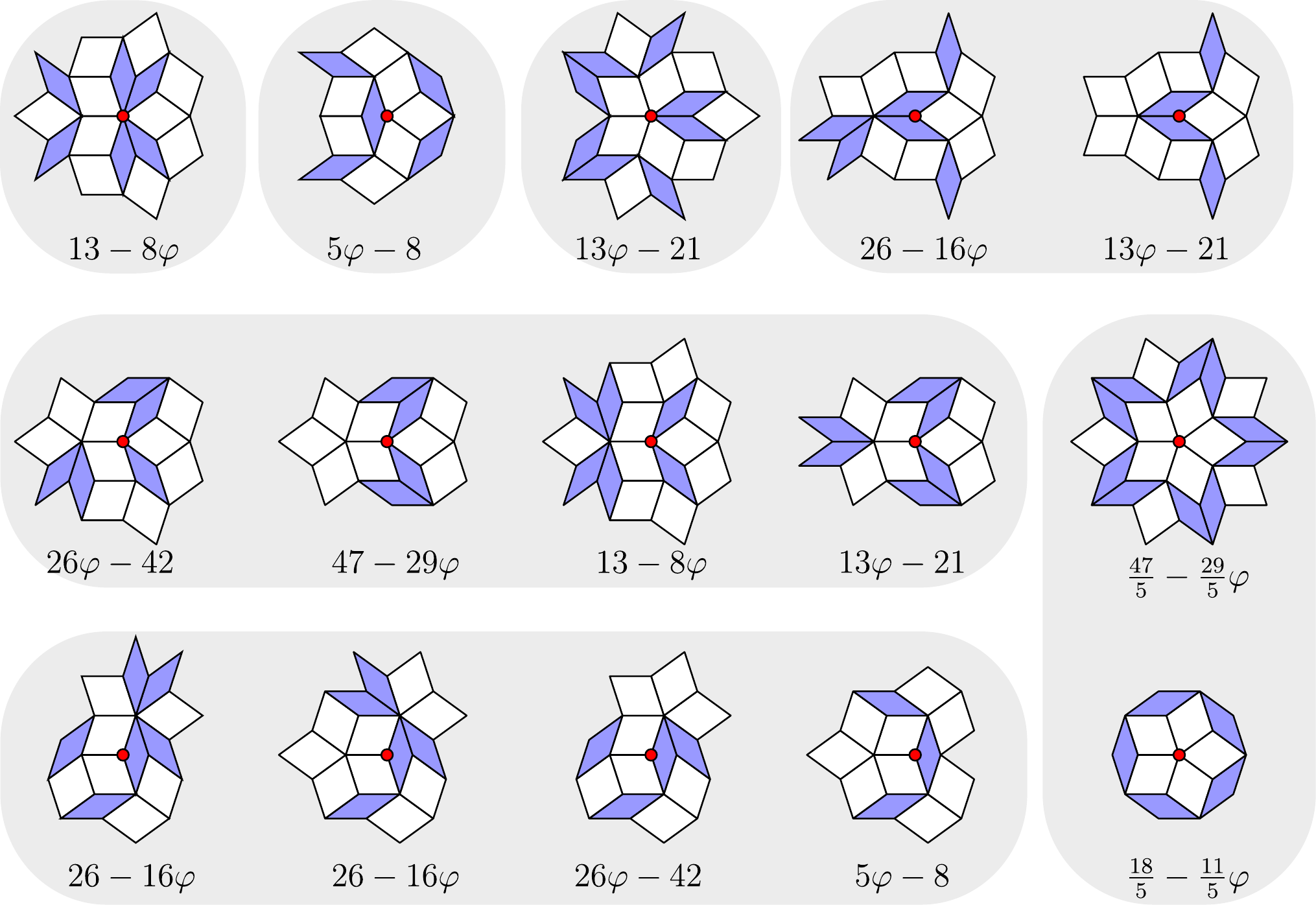}
  \caption{The $1$-atlas of Penrose tilings $\atlas_1$ (up to isometry, grouped by $0$-maps), the numbers are the frequencies of the patterns with $\phi=(1+\sqrt{5})/2$ the golden ratio.\\
    %% Any tiling whose $1$-maps all belong to this list (up to isometry) is a geometrical Penrose tiling.
  }
  \label{fig:1_atlas}
\end{figure}

This can be proved in many different ways.
The natural idea would be to perform a brute-force combinatorial exploration of all the $1$-maps with labeled tiles and then to remove the labels. 
However this would not be sufficient because there exist \emph{deceptions} \cite{dworkin1995}, \emph{i.e.}, finite labeled patterns that satisfy the matching rules induced by the labels but that do not appear in any infinite Penrose tiling, see Figure \ref{fig:deceptions}.
Some deceptions might not be computationally eliminated as the general problem of deciding whether a pattern is globally allowed is undecidable \cite{robinson1971}.
Note however that, for the specific case of Penrose tilings, Sections \ref{subsec:substitution}, \ref{subsec:graph_substitution} and \ref{subsec:cut-and-project} provide three algorithms to decide if a given finite pattern is globally allowed.

\begin{figure}[htp]
  \centering
  \includegraphics[width=0.3\textwidth]{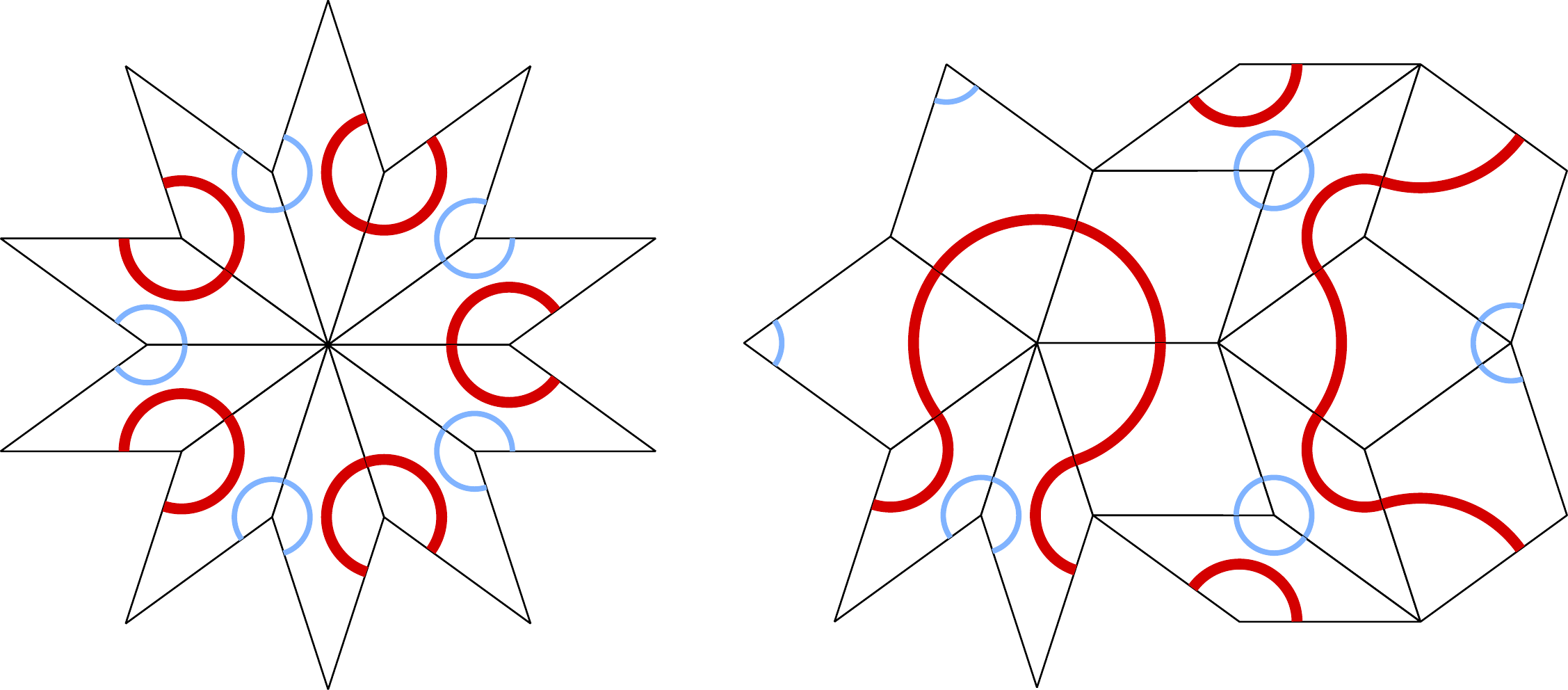}
  \caption{Two simple examples of \emph{deceptions}: patterns that are allowed by the local rules but that cannot be extended to full tilings of the plane.}
  \label{fig:deceptions}
\end{figure}

Rather than this brute-force combinatorial exploration, we present three methods, the first and second one use the substitutive definition of Penrose tilings and the third one uses the cut-and-project definition of Penrose tilings. Note that the third method also gives frequencies of appearance for each $1$-map. 

\subsection{Finding the $1$-atlas using the substitution and linear recurrence}
\label{subsec:substitution}
In this method we use the \emph{substitutive} property of Penrose tilings, to obtain the \emph{linear recurrence} of Penrose tilings with an explicit upper bound for the \emph{ratio of linear recurrence}, hence we can obtain a finite region a Penrose tiling that contains all the $1$-maps, see Fig. \ref{fig:patch_of_apparition}.

\begin{figure}[htp]
  \center
  \includegraphics[width=\textwidth]{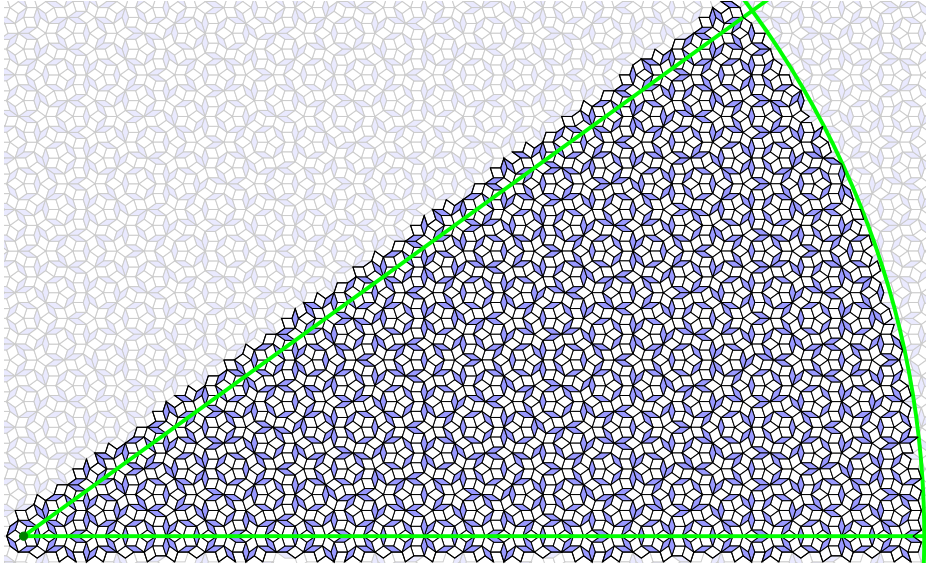}
  \caption{The patch $P_{\atlas_1}$ in which all the $1$-maps of the geometrical Penrose tilings appear, in green the cone of angle $\pi/5$ and the circle of radius $86.57$.}
  \label{fig:patch_of_apparition}
\end{figure}

A \emph{substitution} $\sigma$ is an inflation-subdivision function that maps each tile to a patch of tiles \cite{grunbaum1987}, see Fig.~\ref{fig:penrose_substitution}. It is called \emph{vertex-hierarchic} when there exists a linear map called \emph{expansion} $\phi$ such that for any tile $t$ the vertices of $\phi(t)$ are boundary vertices of the patch of tiles $\sigma(t)$, in Fig.~\ref{fig:penrose_substitution} the expansion is drawn in dotted lines.

Here we are interested in the case of the expansion being a direct similarity, we consider the expansion as the multiplication by a scalar and we use the same symbol for the scaling factor and the funciton of multiplying by it, \emph{i.e.} , we write the expansion $\phi(t)$ as $\phi\cdot t$ where $\phi$ is the expansion factor.

\begin{figure}[h]
  \centering
  \includegraphics[width=0.7\textwidth]{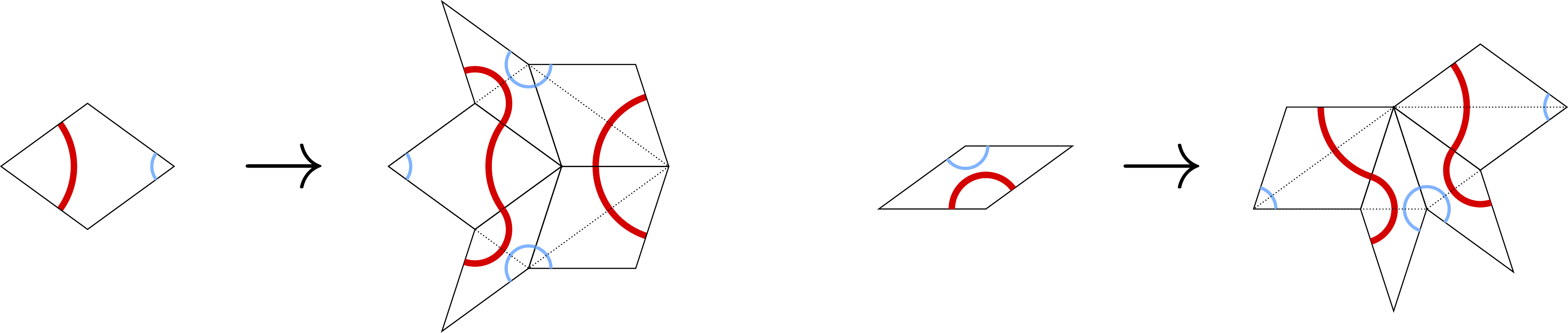}
  \caption{The Penrose substitution $\sigma$ \cite{penrose1974}.}
  \label{fig:penrose_substitution}
\end{figure}

A substitution $\sigma$ is called \emph{primitive} when there exists an integer $k$ such that for any tile $t$, the patch $\sigma^k(t)$ contains all the tiles in the tileset (in every orientation).

A substitution $\sigma$ defines the tiling space $X_\sigma$ as the set of tilings $\tiling$ such that any finite patch $P\in\tiling$ appears in some $\sigma^k(t)$ for $k\in\mathbb{N}$ and $t$ a single tile. Tilings in $X_\sigma$ are called \emph{$\sigma$-tilings}.

Given a  pattern $P$ and a tiling space $X$ (or a single tiling $\tiling$), we call \emph{appearance radius} of $P$ the infimum of the radiuses $r$ such that the pattern $P$ appears in any disc of radius $r$ in any tiling in $X$.

When the appearance radius of every pattern is finite, the tiling or set of tilings is called \emph{uniformly recurrent} or \emph{uniformly repetitive}.
When, additionally, there exists a constant $C$ such that for any pattern $P$ the appearance radius of $P$ is at most the radius of $P$ multiplied by $C$, then the tiling or set of tilings is called \emph{linearly recurrent} with \emph{recurrence factor} $C$.

\begin{lemma}[\cite{solomyak1998}, Linear Recurrence]
  \label{lemma:solomyak}
  Let $\mathbf{T}$ be a set of polygonal tiles. Let $\sigma$ be a primitive vertex-hierarchic substitution with the expansion being a similarity of scaling factor $\phi>1$. \\
  Then any $\sigma$-tiling is linearly recurrent with recurrence factor at most
  \[ C := \frac{\phi C_0}{C_1}\]
  where $C_0$ is the appearance radius of the 0-maps in the $\sigma$-tiling and $C_1$ is a radius such that any disc of radius $r<C_1$ in a tiling by $\mathbf{T}$ tiles is entirely covered by a 0-map.
\end{lemma}

\begin{remark}
  Primitive substitution tilings are uniformly recurrent, so $C_0$ exists.
  
  $C_1$ is defined as any radius such that every disc is covered by a $0$-map, in order to minimize the recurrence factor $C$ we can take $C_1$ as the supremum of the suitable radiuses.
  
  In \cite{solomyak1998} the lemma is stated for self-similar tilings, \emph{i.e.}, regular tilings for an edge-hierarchic substitution (the union of the tiles in $\sigma(t)$ is exactly the expanded tile $\phi(t)$) whose expansion is a similarity, however the result generalizes in a straightforward way to the more general vertex-hierarchic substitutions because the key element of the proof is that the substitution is an expansion-subdivision process with the expansion a similarity of scaling factor $\phi > 1$.
	
  The hypothesis of the expansion being a similarity is necessary, see Appendix \ref{appendix:solomyak}. 
	
  As we are interested in the $1$-atlas up to isometry, we consider the "up-to-isometry" recurrence factor, \emph{i.e.}, a factor $C$ such that for any pattern $P$ that appears in a Penrose tiling, $P$ appears up to isometry in any disc of radius $C\cdot \mathrm{radius(P)}$.
\end{remark}                          

\begin{lemma}[Linear recurrence factor for Penrose tilings]
  The Penrose  rhombus tilings are linearly recurrent with a recurrence factor up to isometry at most $C_p$ with
  \[ C_p:= \frac{\phi\cdot(1+\phi + \sqrt{19 + 30\phi})}{2\cos(\tfrac{2\pi}{5}) \sin(\tfrac{2\pi}{5})} < 29.830 . \]
\end{lemma}

\begin{figure}[ht]
  \centering
  \includegraphics[width=0.8\textwidth]{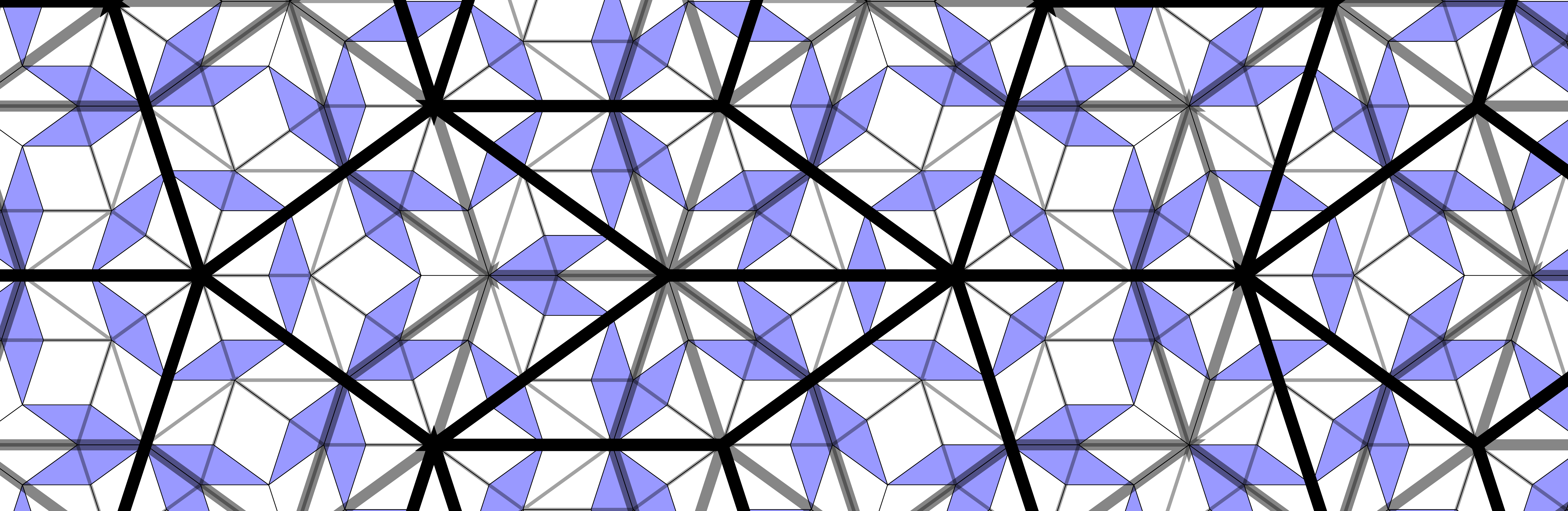}
  \caption{The decomposition of a geometrical Penrose tiling in metatiles of order $1$ (thin grey lines), $2$ (bold grey lines) and $3$ (bold black lines).}
  \label{fig:third_decomposition}
\end{figure}
\begin{remark}
  Note that this statement and proof concern labelled (arrowed or with coloured arcs) Penrose rhombus tilings. 
  They imply the same result for Geometrical Penrose rhombus tilings.	
\end{remark}

\begin{proof}
  Let us first recall that Penrose tilings are substitution tilings with the substitution $\sigma$ of Figure \ref{fig:penrose_substitution} \cite{penrose1974}  which is vertex-hierarchic and has scaling factor $\phi = \tfrac{1+\sqrt{5}}{2}$, note that $\phi$ is the golden ratio and that $\phi^2 = \phi+1$. 
  
  In order to apply Lemma \ref{lemma:solomyak} we now compute $C_0$ and $C_1$ for Penrose tilings.
  $C_1$ is precisely the inner radius of the thin rhombus tile, 
  \[C_1 = 2\cos(\tfrac{2\pi}{5}) \sin(\tfrac{2\pi}{5}) \approx 0.588.\]
  Indeed a circle of radius more than $C_1$ centered on the center of a thin Penrose rhombus overlaps on all four sides of the tile and hence is not covered by a single $0$-map. 
  And a circle of radius at most $C_1$ is either completely covered by a tile, overlaps along two adjacent edges or overlaps along only one edge. In all three cases it is covered by a single $0$-map.
  
  $C_0$ is harder to compute. Recall that we are interested in the appearance radius of the 0-maps in the labelled Penrose rhombus tilings. There are 8 labelled 0-maps (Fig.~\ref{fig:0_atlas_labels}).
  
  \begin{figure}[ht]
    \center
    \includegraphics[width=0.6\textwidth]{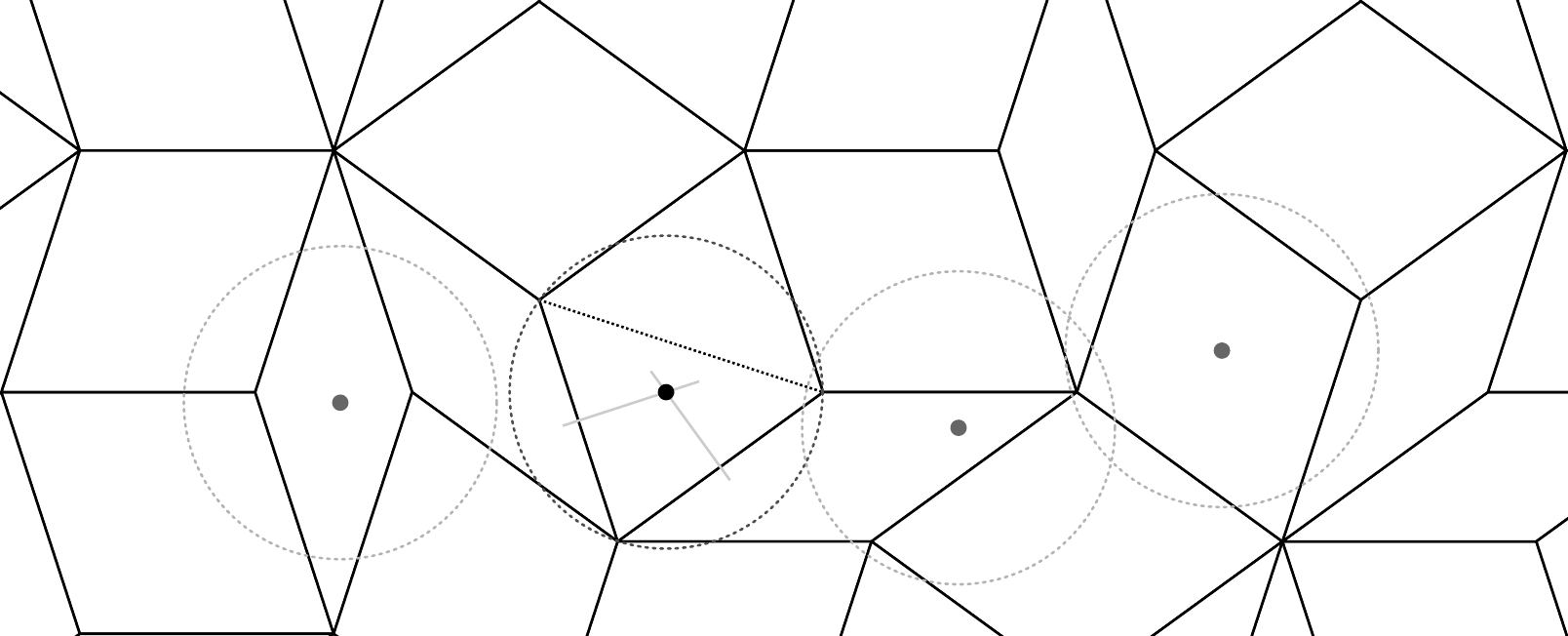}
    \caption{The maximal distance $r_v'$ from a point of $\mathbb{Z}^2$ to a vertex of a Penrose tiling is the radius of the circle circumscribed to the half fat rhombus.}
    \label{fig:r_v'}
  \end{figure}
  We give an upper and lower bounds $r_c \leq C_0 \leq r_c+r_v$ with $r_c:= \sqrt{19 + 30\phi}$ and $r_v:= 1+\phi$. To prove these bounds we again use the Penrose substitution : any Penrose tiling $\tiling$ is the image of a Penrose tiling $\tiling_{-1}$ by the substitution, which itself is also the image of a Penrose tiling by the substitution.
  Let us take a Penrose tiling $\tiling$, there exists a Penrose tiling $\tiling_{-3}$ such that $\tiling = \sigma^3(\tiling_{-3})$. 
  This induces a decomposition of $\tiling$ in \emph{metatiles} of order 3, \emph{i.e.}, the image by $\sigma^3$ of the tiles in $\tiling_{-3}$ (Fig.~\ref{fig:third_decomposition}). 
  \begin{figure}[htp]
    \center
    \includegraphics[width=0.6\textwidth]{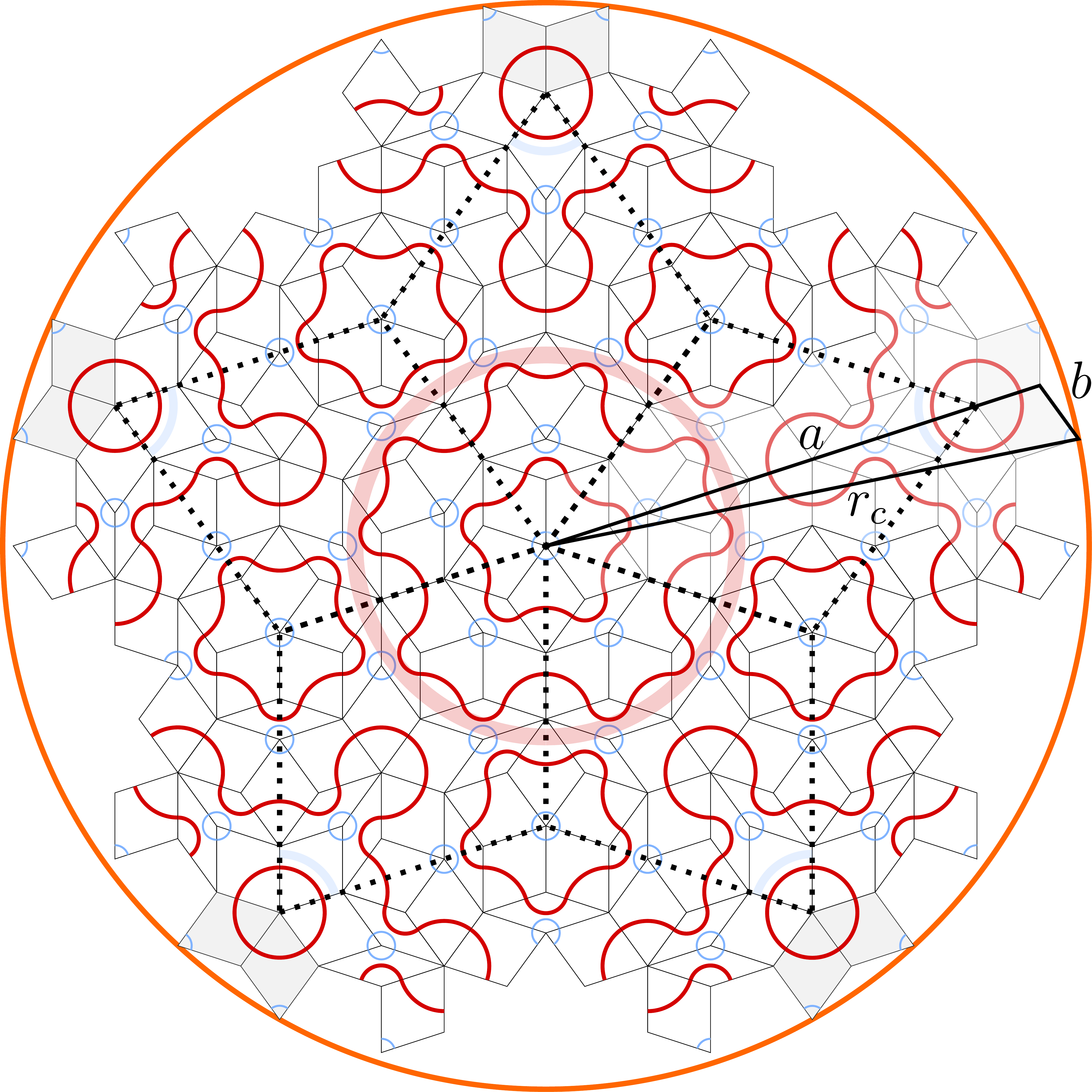}
    \caption{The image by $\sigma^3$ of a star pattern. All the $0$-maps appear in the circle of radius $r_c$.}
    \label{fig:sigma3_0map_bound}
  \end{figure}

  Let us denote by $r_v$ the maximal distance from a point of $\mathbb{R}^2$ to the closest corner of a metatile of order 3 and $r_v'$ the maximal distance from a point of $\mathbb{R}^2$ to the closest vertex in a Penrose tiling. 
  We have $r_v = \phi^3r_v'$. We also have that $r_v'$ is exactly the radius of the circumscribed circle to the triangle consisting of the fat Penrose rhombus bisected along its short diagonal (Fig.~\ref{fig:r_v'}), so we have
  \[ r_v = \phi^3 r_v' = \frac{\phi^3}{2\sin(3\pi/10)} = \phi^2 = 1 + \phi.\]
  
  %%\pagebreak
  
  All the $0$-maps (up to isometry) appear in $\sigma^3(P)$ for any $0$-map $P$, see Fig.~\ref{fig:sigma3_0maps}.
  We denote by $r_c$ the appearance radius of the $0$-maps up to isometry around the center of third order metatiles $0$-maps, \emph{i.e.}, the image of $0$-maps by $\sigma^3$, see Fig.~\ref{fig:sigma3_0maps}.
  This bound is reached by the appearance radius of the star pattern around the third image by $\sigma$ of a star pattern, see Fig.~\ref{fig:sigma3_0map_bound}. This yields 
  $r_c=\sqrt{a^2 + b^2 - abc}$ with \[ a := 3(1+\phi)\qquad b:= 1 \qquad c = 2\cos\tfrac{3\pi}{5} = 1 - \phi, \] we simplify $\phi^2 = \phi + 1$  and we obtain $r_c = \sqrt{19 + 30\phi}$.
  
  We obtain the claimed bounds $r_c \leq C_0 \leq r_c + r_v$, indeed the appearance radius of the $0$-maps around a point of $\mathbb{R}^2$ is at most the distance from this point to a corner of third order metatile plus the appearance radius around this corner of third order metatile.
  
  Now we apply Lemma \ref{lemma:solomyak} to obtain the expected bound $\phi C_0/ C_1 \leq \phi (r_c + r_v)/C_1$.
\end{proof}
\begin{figure}[htp]
  \center
  \includegraphics[width=0.8\textwidth]{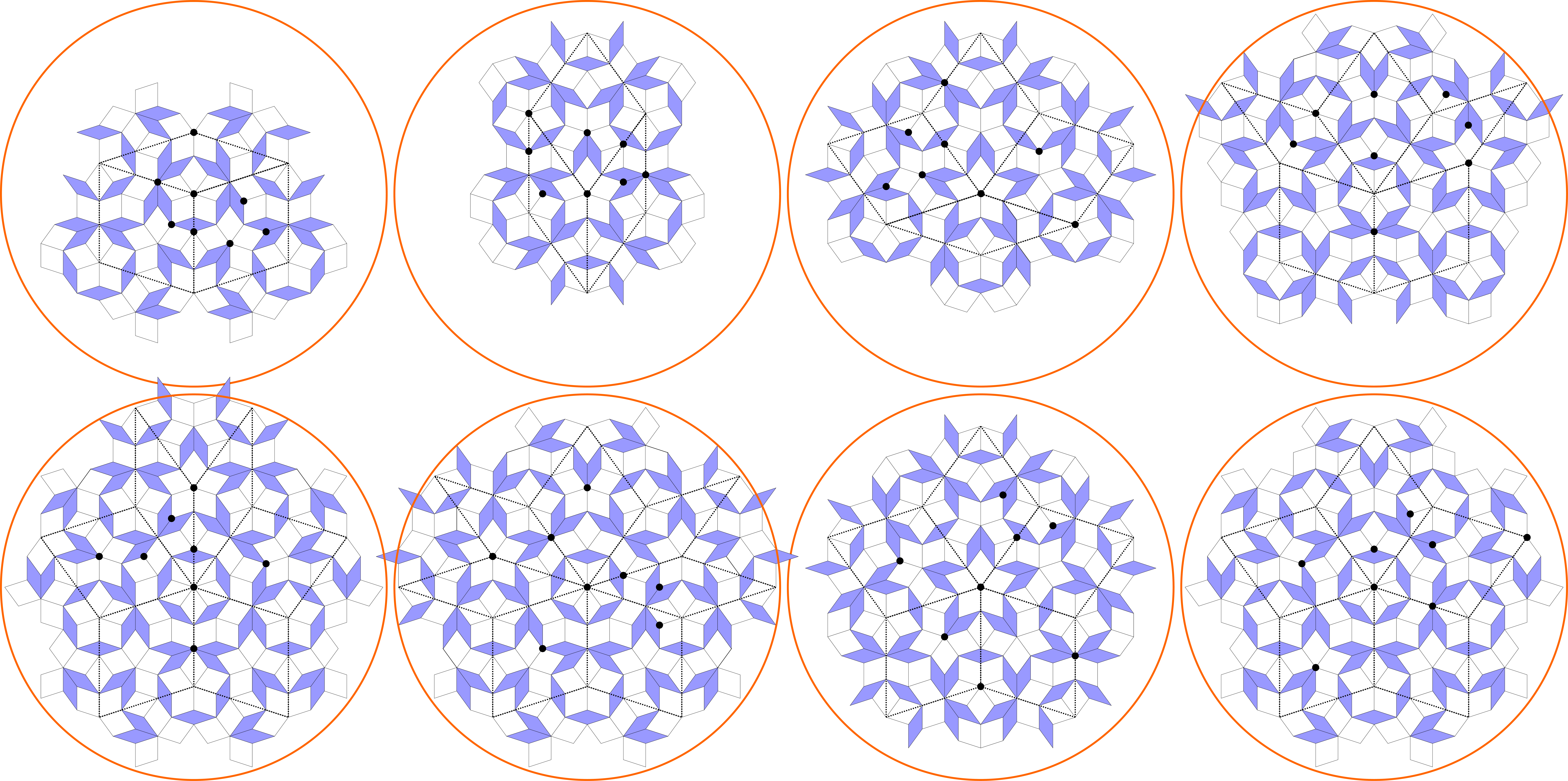}
  \caption{The image by $\sigma^3$ of the geometrical Penrose $0$-maps (in the same order as in Fig.~\ref{fig:0_atlas_labels}), in each patch there is a copy of each $0$-maps up to isometry that lies inside the circle of radius $r_c$ (in orange), one occurence of each pattern is marked with a black dot.
    Note that the sun/star $0$-map must appear in both orientations, so there are 8 $0$-maps marked in each $\sigma^3(P)$.
    Note also that the star pattern appears on the very boundary of $\sigma^3(\mathrm{kite})$ and of $\sigma^3(\mathrm{star})$ but that the missing rhombus is forced by the neighbour third order metatile.
  }
  \label{fig:sigma3_0maps}
\end{figure}

\begin{figure}[htp]
  \center
  \includegraphics[width=\textwidth]{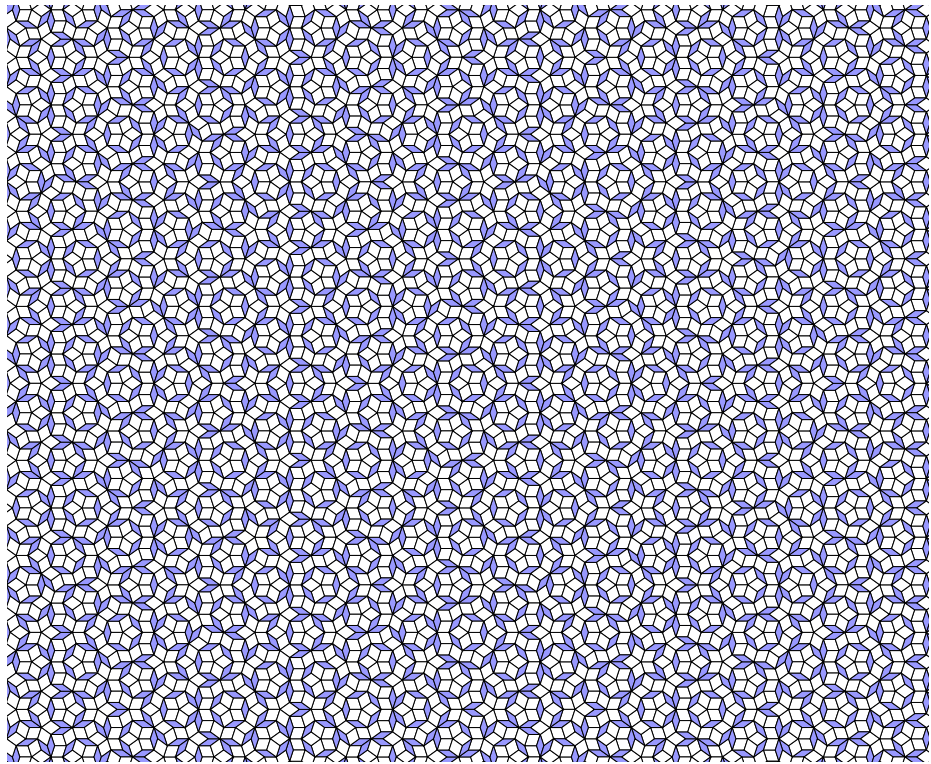}
  \caption{A central fragment of the (geometrical) canonical Penrose tiling, we can observe its 5-fold rotational symmetry around the origin and its horizontal reflexion symmetry.}
  \label{fig:canonical_penrose}
\end{figure}

Remark that the radius of the $1$-maps is bounded by $1+2\cos(\pi/10)$ as the maximum diameter of a Penrose rhombus is $2\cos(\pi/10)$. 
Which means that, by linear recurrence, all the $1$-maps must appear up to isometry in any patch of radius $R_{\atlas_1} = (1+2\cos(\pi/10))C < 86.57$.

Recall that the local rules or the substitution define the subshift of Penrose tilings and not a single tiling.
However we often want to refer to a single Penrose tiling, which we then call \emph{canonical Penrose tiling}.
This tiling can be defined as the fixpoint obtained by iterating $\sigma^4$ on the sun pattern \cite{penrose1974} (indeed the sun pattern is at the center of its fourth image by $\sigma$), as the dual tiling of the pentagrid of offset $\tfrac{1}{5}$ \cite{debruijn1981} or as the Penrose cut-and-project tiling of intercept $\tfrac{1}{5}$ \cite{baake2013} (see Section \ref{subsec:cut-and-project}).
The canonical Penrose tiling has global $5$-fold rotational symmetry around the origin, and a reflexion symmetry along the horizontal axis. See Figure \ref{fig:canonical_penrose} for a central fragment of the canonical penrose tiling.

All the Penrose $1$-maps appear up to isometry in the central disc of radius $R_{\atlas_1}$ in the canonical Penrose tiling. By symmetry of the tiling we can reduce the disc to a $\pi/5$ cone of the disc. 
We define $P_{\atlas_1}$ the patch of the canonical Penrose tiling of tiles that are in the central disc of radius $R_{\atlas_1}$ and are at edge distance at most $1$ of the cone of angle  $\pi/5$, see Fig.~\ref{fig:patch_of_apparition} and \ref{fig:patch_of_apparition_scattered}. The $1$-atlas up to isometry appears in $P_{\atlas_1}$.
We enumerate all the $1$-maps in this finite pattern and we obtain the $1$-atlas presented in Figure \ref{fig:1_atlas}, and with pointed occurences of the patterns in Figure \ref{fig:patch_of_apparition_scattered}.
\begin{figure}[!htp]
  \includegraphics[width=\textwidth]{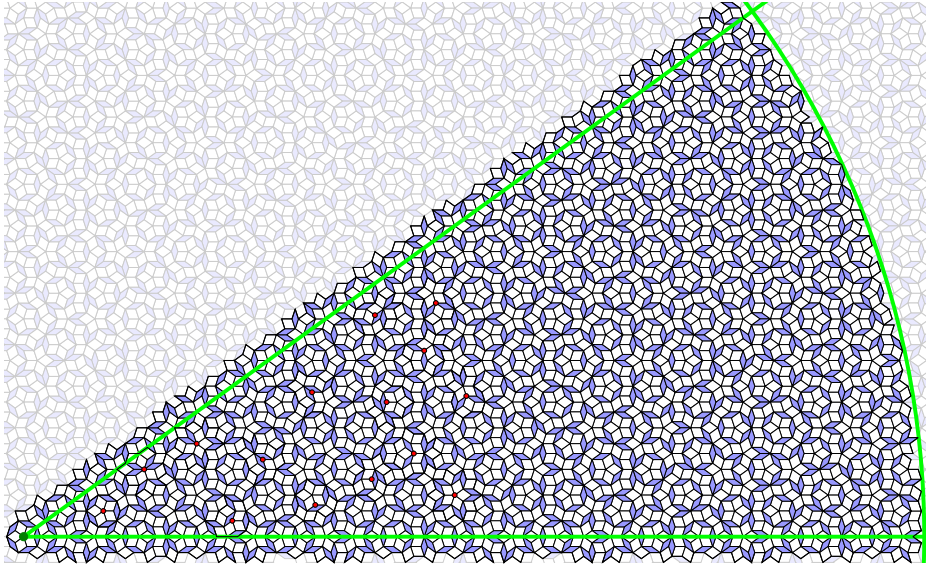}
  \caption{The patch $P_{\atlas_1}$ with scattered occurrences of the 1-maps highlighted.}
  \label{fig:patch_of_apparition_scattered}
\end{figure}

Note however that the $15$ $1$-maps appear in a very small sub-patch of $P_{\atlas_1}$, see Figure \ref{fig:patch_of_apparition_minimal}.
Overall, the ratio of linear recurrence that we prove is not optimal, however it is hard to find an optimal value because the asymptotic ratio of appearance radius by the radius of the pattern is 1 because Penrose tilings are densely repetitive \cite{lenz2002}.
So the bound is not asymptotic but is reached for some finite pattern (probably of relatively small size) and finding it would require a very extensive combinatorial exploration and then a computer assisted proof.

\begin{figure}[!htp]
  \center
  \includegraphics[width=0.5\textwidth]{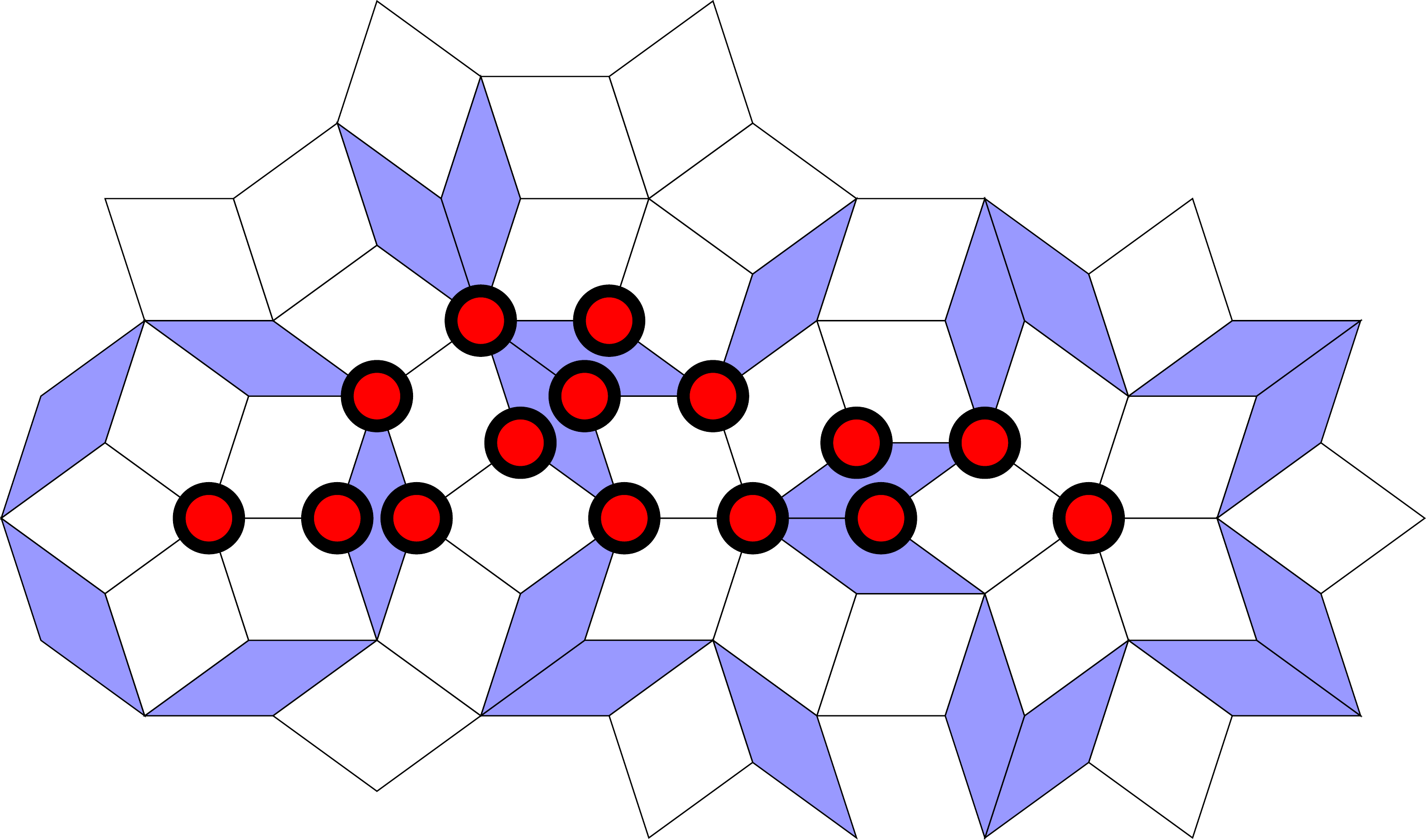}
  \caption{A minimal subpatch of $P_{\atlas_1}$ containing all the 15 $1$-maps, the center of the $1$-maps are highlighted.}
  \label{fig:patch_of_apparition_minimal}
\end{figure}

\pagebreak

\subsection{Finding the $1$-atlas with the substitution graph}
\label{subsec:graph_substitution}

Another way to prove that the $1$-atlas $\atlas_1$ presented in Fig.~\ref{fig:1_atlas} is correct is to build the directed graph $G=(V,E)$ where $V$ is the set of $1$-maps up to isometry and $u\to v\in E$ when $v$ appears in $\sigma(u)$.

Since the Penrose substitution is primitive, we can start from a single pattern, for example take the $1$-map that is the sun pattern surrounded by a layer of thin rhombi, apply the substitution to obtain new vertices/$1$-maps and repeat until no new $1$-map is generated, see Algorithm \ref{algo:atlas_subst}.

\begin{figure}[htp]
  \includegraphics[width=\textwidth]{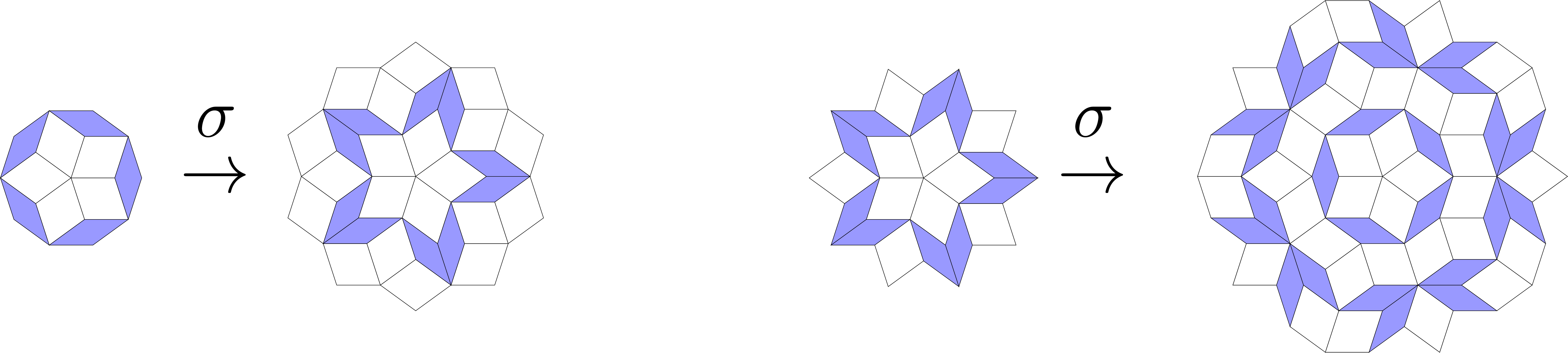}
  \caption{Applying the substitution on a sun $1$-map (left), we obtain only two $1$-maps up to isometry : a central star and a corolla of rotated star-jack patterns.
  Applying the substitution on a star $1$-map (right), we obtain five one-maps including a sun. }
  \label{fig:sun_subst}
\end{figure}

\begin{algorithm}
	\KwData{a well-defined primitive susbtitution $\sigma$ and a $k$-map $P_0$}
	\KwResult{the list of the $k$-maps of the $\sigma$-tilings up to translation}
	$S\gets [P_0]$\;
	$V\gets [P_0]$\;
	$E\gets []$\;
	\While{$S\neq \emptyset$}{
	$P\gets$ extract an element from $S$\;
	\For{$P' \in \mathrm{1maps}(\sigma(P))$}{
	\If{ $P' \notin V$}{
	append $P'$ to $V$\;
	append $P'$ to $S$\;
	}
	append the directed edge $(P,P')$ to $E$\;
	}
	}
	\KwRet $V$\;
	\caption{Computing the $k$-atlas with a primitive substitution $\sigma$ and the corresponding graph.}
	\label{algo:atlas_subst}
\end{algorithm}

When we apply this algorithm (with the up-to-isometry equivalence on $1$-maps) we obtain the graph depicted in Fig.~\ref{fig:graph_subst}, and the adjacency matrix of Fig.\ref{fig:adjacency_graph_subst}. In particular we obtain as expected the $1$-atlas $\atlas_1$.

\begin{figure}[htp]
  \center
  \includegraphics[width=\textwidth]{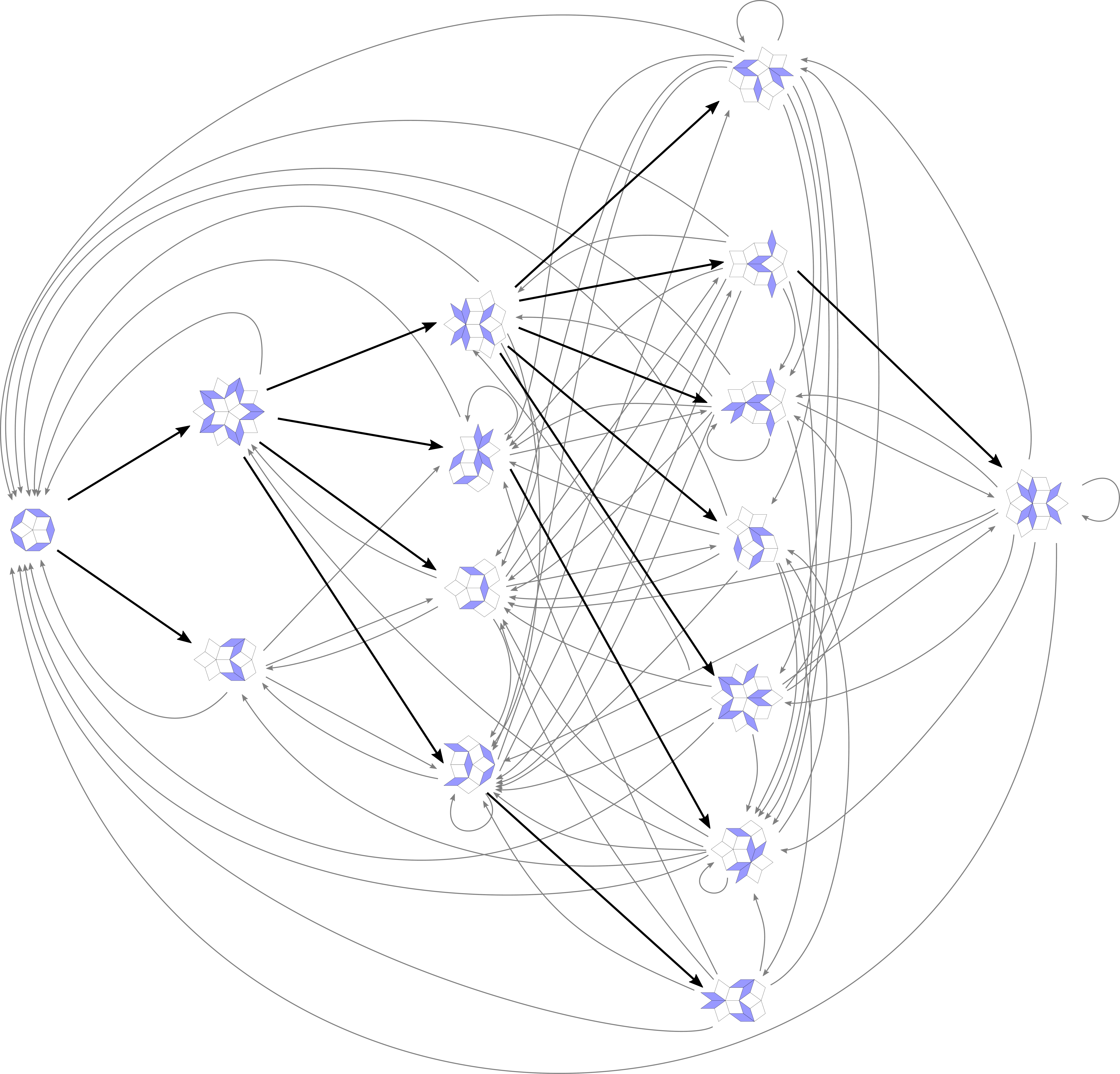}
  \caption{The graph $G=(V,E)$ of the $1$-maps of Penrose tilings under the Penrose substitution $\sigma$, starting Algorithm \ref{algo:atlas_subst} from the leftmost pattern, a bold black arrow indicates a new $1$-map and grey arrows indicate a $1$-map that was already explored.}
  \label{fig:graph_subst}
\end{figure}

\begin{figure}[htp]
  \center
  \includegraphics[width=\textwidth]{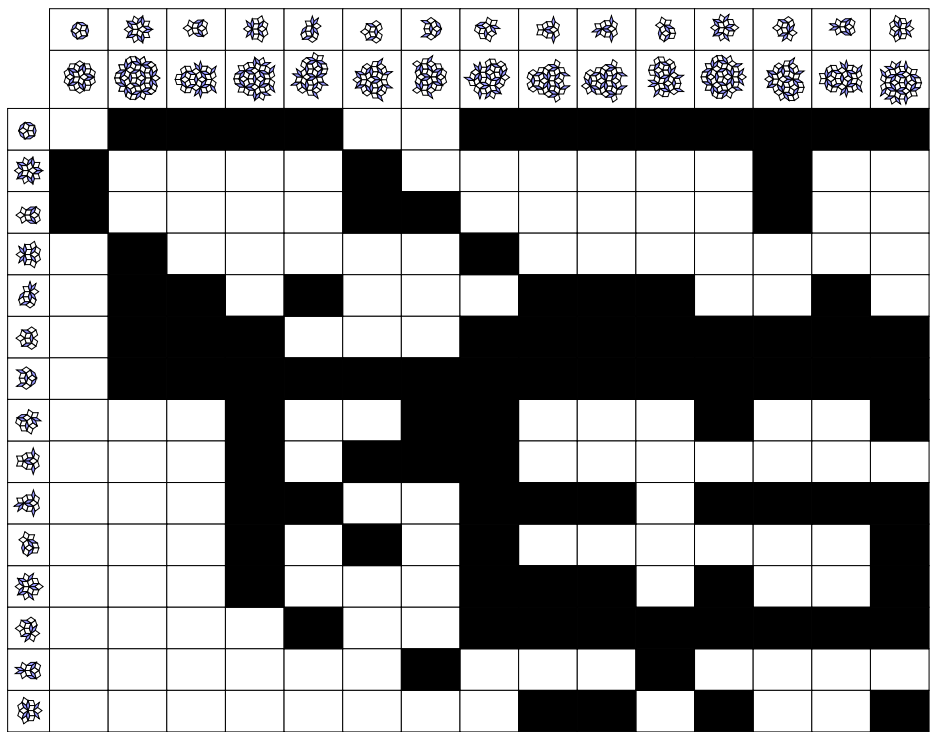}
  \caption{An "adjacency matrix"-like representation of the graph $G=(V,E)$ of the $1$-maps of Penrose tilings under the Penrose substitution $\sigma$, the second line of column header are the images of the $1$-maps by $\sigma$.}
  \label{fig:adjacency_graph_subst}
\end{figure}

\pagebreak

\subsection{Finding the $1$-atlas using cut-and-projection}
\label{subsec:cut-and-project}
This method uses the cut-and-project property of Penrose tilings as explained in \cite{fernique2022}, using the notions of {\em window} and {\em region} of a pattern.

\begin{definition}
  Let $E$ be a $d$-dimensional affine plane in $\mathbb{R}^n$ which does not contain any line directed by an integer vector and such that the boundary of $E+[0,1]^n$ does not intersect $\mathbb{Z}^n$.
  Let $\pi$ denotes the orthogonal projection onto $E$.
  The {\em cut and project} tiling with {\em slope} $E$ is defined as follows:
  \begin{enumerate}
  \item select (cut) the $d$-dim. unit facets of $\mathbb{Z}^n$ which lie inside $E+[0,1]^n$;
  \item project them under $\pi$ to get a tiling of $E$.
  \end{enumerate}
\end{definition}

For example, the cut and project tilings whose slope is a $2$-plane of $\mathbb{R}^5$ directed by the two vectors of $\mathbb{R}^5$
\[
(\cos\tfrac{2k\pi}{5})_{0\leq k < 5}
\qquad\textrm{and}\qquad
(\sin\tfrac{2k\pi}{5})_{0\leq k < 5}
\]
are called {\em generalized Penrose tilings} \cite[\S 6.4.3]{senechal1996}.
Among them, the Penrose tilings correspond to affine slopes which contain a point whose coordinate sum up to an integer \cite{debruijn1981, senechal1996, baake2013}.

The fact that the cut and project method indeed defines a tiling is proven in \cite{debruijn1986}.
The selected facets actually form a $d$-dim. surface and the projection $\pi$ is a homeomorphism between this surface and $E$.
The {\em lift} of a vertex $x$ of the tiling, denoted by $\widehat{x}$, is the unique point of this surface which projects onto $x$.

\begin{definition}
  Let $\pi'$ denote the orthogonal projection onto the orthogonal complement $E'$ of $E$.
  The {\em window} of a cut and project tiling with slope $E$ is the $(n-d)$-dimensional polytope
  \[
  W:=\pi'(E+[0,1]^n).
  \]
\end{definition}

For $x'\in W$ and $k\in\mathbb{N}$, define the set $B(x',k)$ of points $u$ of $\mathbb{Z}^n$ of 1-norm at most $k$ and such that $x'+\pi'u$ is still in the window $W$, \emph{i.e.}, 
\[B(x',k):=\{u\in\mathbb{Z}^n,~||u||_1\leq k,~x'+\pi'u\in W\}\]
and denote by $V(x',k)$ the vertices in $B(x',k+2)$ which have at least two neighbors in $B(x',k+1)$, where $u$ and $v$ are neighbors if $||u-v||_1=1$.
\begin{remark}
  \label{rem:kmap_vertices}
%% 	The vertices of a $k$-map, denoted by $V(x,k)$, are the vertices at distance at most $k+2$ from $x$, denoted by $B(x,k+2)$, which have at least two neighbors at distance at most $k+1$ from $x$.
        
  This holds because the vertices of a $k$-map centered in $x$ are exactly the union of:
  \begin{itemize}
  \item the vertices within distance $k$ from $x$;
  \item the vertices at distance $k+1$ from $x$ (they all belong to a tile which has a vertex at distance $k$);
  \item the vertices at distance $k+2$ from $x$ that have two neighbors at distance $k+1$ from $x$ (they belong to a tile with one vertex at distance $k$ from $x$, two at distance $k+1$ and one at distance $k+2$), see Fig.\ref{fig:kmap_vertices}.
  \end{itemize}
\end{remark}
Hence, if there is a vertex $x$ of the tiling such that $\pi'\widehat{x}=x'$, then $x+\pi V(x',k)$ is exactly the set of vertices of the $k$-map centered on $x$. 
This definition however holds for any $x'\in W$, not only for the countably many which are in $\pi'\mathbb{Z}^n$.
Now, define
\[
R(x',k):=\bigcap_{u\in V(x',k)}\left(W-\pi'u\right).
\]
This is a polytope included in $W$, called a {\em $k$-region} of the tiling.

From these definitions we obtain that the image by $\pi'$ of an integer point $\widehat{y}$ in $E+[0,1]^n$ is in $R(x',k)$ iff its image by $\pi$ is the center of a $k$-map with vertices $\pi\widehat{y}+\pi V(x',r)$.

%% \begin{proof}
%% 	The only difficulty is in the formalism:
%% 	\begin{eqnarray*}
%% 		\pi'\widehat{y}\in R(x',k)
%% 		&\Leftrightarrow& \forall~u\in V(x',k),~\pi'\widehat{y}\in \left(W-\pi'u\right)\\
%% 		&\Leftrightarrow& \pi'\left(\widehat{y}+V(x',k)\right)\subset W\\
%% 		&\Leftrightarrow& \widehat{y}+V(x',k)\subset E+[0,1]^n\\
%% 		&\Leftrightarrow& \pi\widehat{y}+\pi V(x',r) \textrm{ are vertices of the tiling.}
%% 	\end{eqnarray*}
%% 	The definition of $V(x',r)$ yields that those are the vertices of a $k$-map.
%% \end{proof}

In other words, there is an explicit bijection between the $k$-maps of the tiling (considered up to translation) and its $k$-regions.
More exactly, we have to consider only the $k$-regions which contain the image by $\pi'$ of an integer point, which is generic (since $\pi'\mathbb{Z}^n$ is generically dense in $E'$).
However, Penrose's case is precisely not generic, as explained in Fig.~\ref{fig:pattern_region}.

\begin{figure}[htp]
  \centering
  \includegraphics[width=0.6\textwidth]{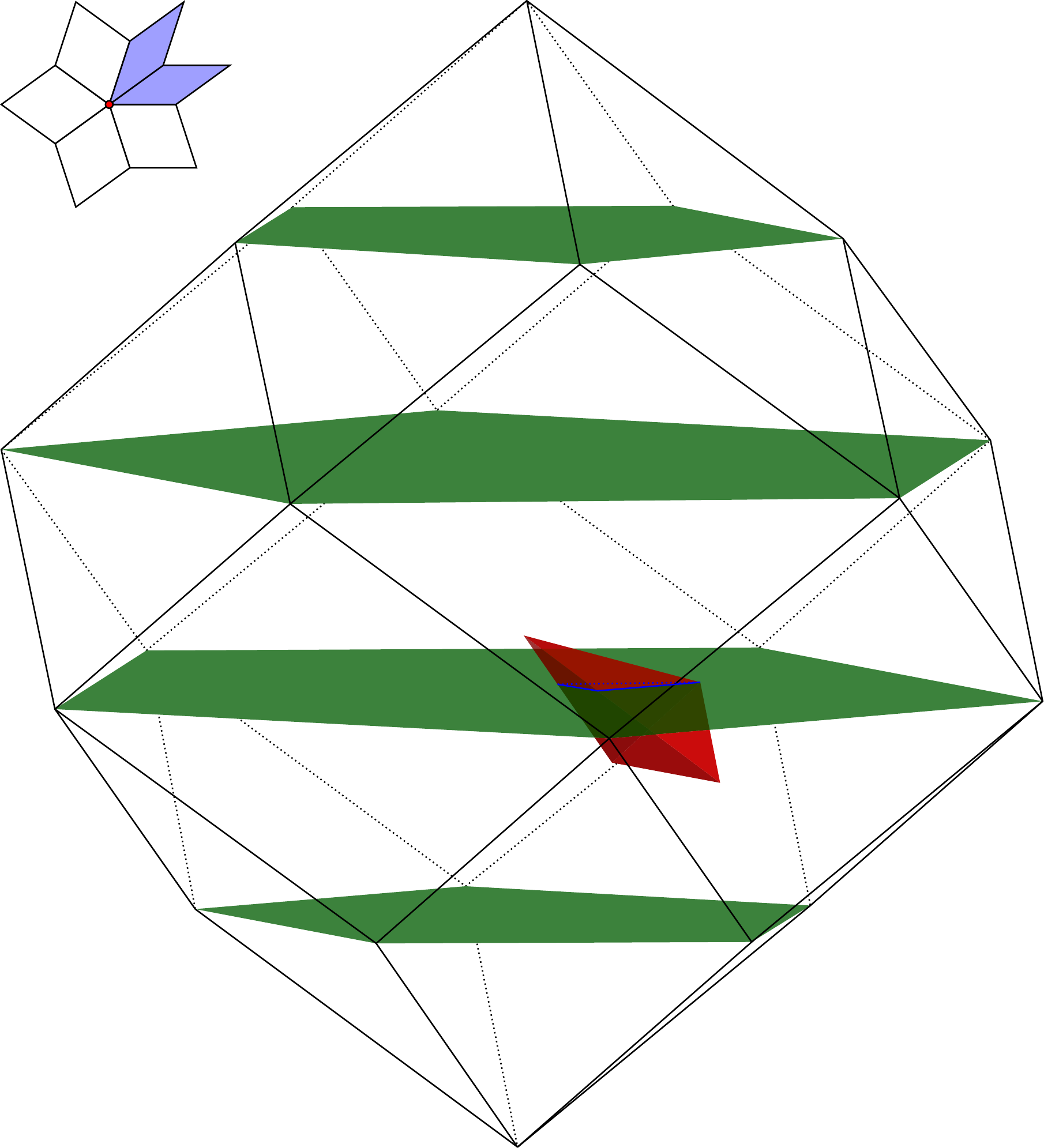}
  \caption{
    The window of a Penrose tiling is a rhombic icosahedron.
    Top-left, a $0$-map is depicted.
    The corresponding region is the red tetrahedron in the window.
    The slope of generalized Penrose tilings is however specific: it is contained in a four dimensional rational subspace of $\mathbb{R}^5$ (namely the space orthogonal to $(1,1,1,1,1)$) and the points of $\pi'\mathbb{Z}^5$ are not dense in the window but form a family of parallel planes (whose intersection with the window is here depicted in green).
    Hence, only the regions which intersect these planes will indeed correspond to patterns, as it is the case here with the red tetrahedron.
    Shifting the slope shifts these planes which may intersect different regions: this means that two generalized Penrose tilings may have different finite patterns.
    Recall that Penrose tilings have been defined as generalized Penrose tilings with a slope which contains a point whose coordinates sum up to $1$.
    This amounts to fix the green parallel planes to go through the vertices of the window, hence to always intersect the same regions.
    This is why Penrose tilings have all the same finite patterns.
  }
  \label{fig:pattern_region}
\end{figure}

Now, the point is that, for every $x'\in W$, both $V(x',k)$ and $R(x',r)$ can be easily computed (it amounts to checking that projections of whole points are in a polytope or intersecting whole translates of polytopes).
And this is done in an exact way with computer algebra if $E$ as well is given in an exact way (for example if it is generated by algebraic vectors).
This leads to Algorithm~\ref{algo:atlas} to compute the $k$-atlas of a tiling.

\begin{algorithm}
	\KwData{$d$-dim. affine plane $E$ in $\mathbb{R}^n$, integer $k$}
	\KwResult{the list of the $k$-maps of the cut and project tiling with slope $E$}
	$A\gets []$\;
	$R\gets \emptyset$\;
	$x'\gets \textrm{ random point in }W$\;
	\While{$R\neq W$}{
	append $V(x',k)$ to $A$\;
	$R\gets R\cup R(x',r)$\;
	$x'\gets \textrm{ random point in }W\backslash R$\;
	}
	\KwRet $A$\;
	\caption{Computing a $k$-atlas}
	\label{algo:atlas}
\end{algorithm}
Applying this algorithm with the slope $E$ of generalized Penrose tilings yields a 0-atlas of $153$ patterns and a 1-atlas of $1705$ patterns. Up to isometry, theses sets respectively reduce to $16$ and $110$ patterns.
These are all the patterns which appear in the generalized Penrose tilings. To obtain the 0-atlas with $7$ patterns depicted in Fig.~\ref{fig:0_atlas} or the 1-atlas with $15$ patterns depicted in Fig.~\ref{fig:1_atlas}, which are the patterns which appear in Penrose tilings only, we need to take into account that the slope $E$ of Penrose tilings is such that the points of $\pi'\mathbb{Z}^n$ are not dense in $W$ but lies in parallel planes: only the regions intersected by these planes indeed correspond to pattern of the tilings (Fig.~\ref{fig:pattern_region}, see also Remark~\ref{rem:frequencies}). 

\begin{remark}
	\label{rem:HK}
	Let us mention that a similar approach is used in \cite{koivusalo2021} to compute the complexity of cut and project tilings, but their regions may be not polytopal and even not connected.
	This is because they do not consider $k$-maps but vertices within Euclidean distance $k$ from a given center.
	This is illustrated in Figure \ref{fig:KW21_remark}.	
\end{remark}

\begin{figure}[htp]
	\centering
	\includegraphics[width=0.5\textwidth]{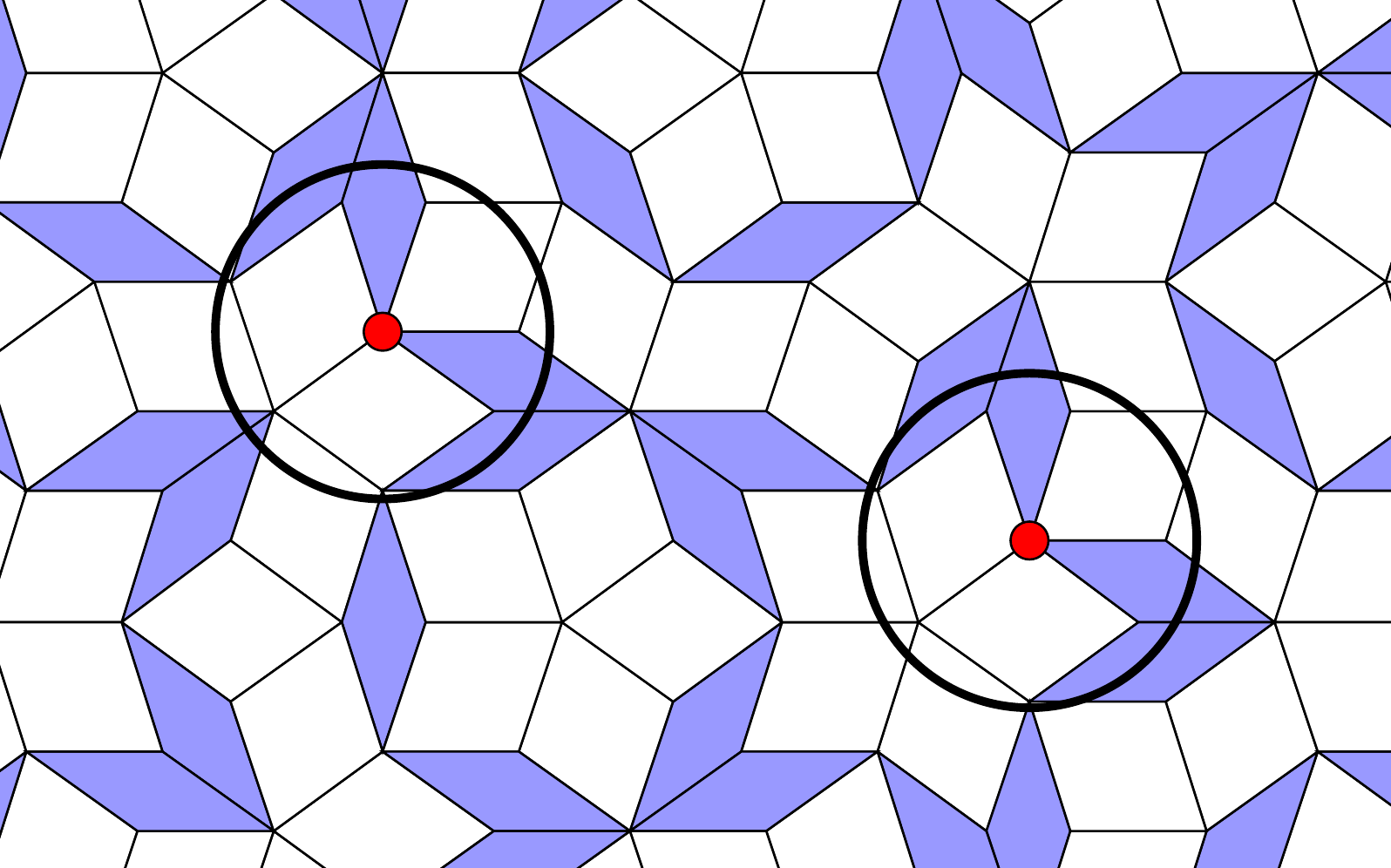}
	\caption{
	The $0$-maps $P_0$ centered around the two red vertices are identical.
	However, the $1$-maps $P_1$ and $P_1'$ centered around the two red vertices are different (two thin blue tiles in the left one correspond to one large white tile in the right one).
	The regions of $P_1$ and $P_1'$ are thus disjoint and may actually be not even adjacent (though they are subset of the region of $P_0$).
	Now, consider the sets of vertices inside these two circles.
	They are identical.
	Hence they correspond to the same region in the sense of \cite{koivusalo2021}.
	But these vertex sets can be extended either as $P_1$ or as $P_1'$, their common region must contains a subset of both the regions of $P_1$ and $P_1'$.
	The region of this vertex set may thus be not connected.
	}
	\label{fig:KW21_remark}
\end{figure}

\begin{remark}
	\label{rem:frequencies}
	We mention that the points of $\pi'\mathbb{Z}^n$ are generically dense in $W$.
	Actually, they are even {\em uniformly distributed}, that is, the proportions of integer points of norm at most $k$ which project by $\pi'$ inside some region in $W$ tends, when $n$ goes to infinity, to the ratio between the volume of this region and the volume of the window.
	In other words, computing these volume ratio yields the frequencies of each $k$-map.
	This is how the frequencies of $0$-maps and $1$-maps of the Penrose tilings (given in Fig.~\ref{fig:1_atlas}) have been computed, with the particularity that in the Penrose case, the points $\pi'\mathbb{Z}^5$ are uniformly distributed not in the whole window $W$ but in parallel planes, so it is necessary to calculate the ratio of the total area of the intersection of a region with the planes and the total area of the intersection of these planes with $W$.
\end{remark}

\section{$\atlas_1$ characterizes Penrose tilings}
\label{sec:characterizes}
We now prove the second part of Theorem \ref{th:main} which we formulate as the following lemma.
\begin{lemma}
	Any tiling by the thin and fat rhombus tiles whose $1$-maps all belong to the atlas $\atlas_1$ depicted in Figure \ref{fig:1_atlas} admits a valid Penrose labelling and therefore is a geometrical Penrose rhombus tiling.
	\label{lemma:labelling}
\end{lemma}
%% XXX QUESTION : should we recall what a geometrical penrose rhombus tiling is? should we put it in a definition environment?

\begin{proof}
	
  Let $\tiling$ be a tiling by the thin and fat rhombus tiles (without coloured arcs) whose $1$-maps belong to the $1$-atlas $\atlas_1$ of Fig.~\ref{fig:1_atlas}, we add coloured arcs to the tiles around each vertex as depicted in Figures~\ref{fig:decorate_tiles} and \ref{fig:decorate_stars}.
	
  \begin{figure}[htp]
    \center
    \includegraphics[width=0.8\textwidth]{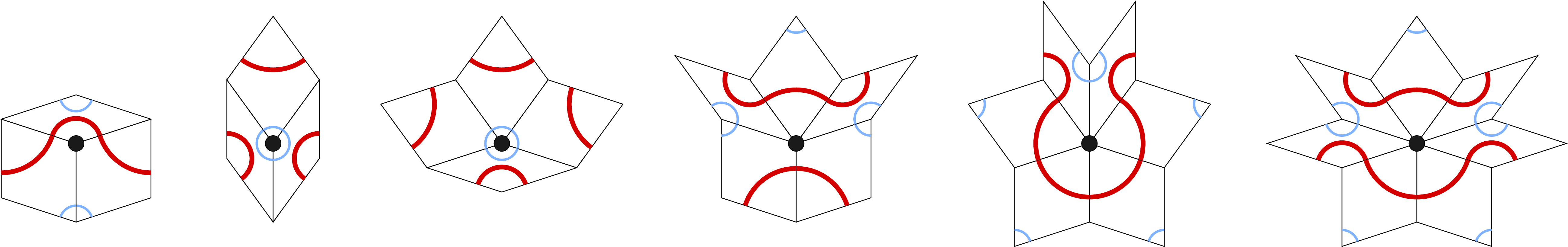}
    \caption{How to decorate the $6$ first $0$-maps.
      See Fig.~\ref{fig:decorate_stars} for the last case.
    }
    \label{fig:decorate_tiles}
  \end{figure}
  
  \begin{figure}[htp]
    \center
    \includegraphics[width=0.35\textwidth]{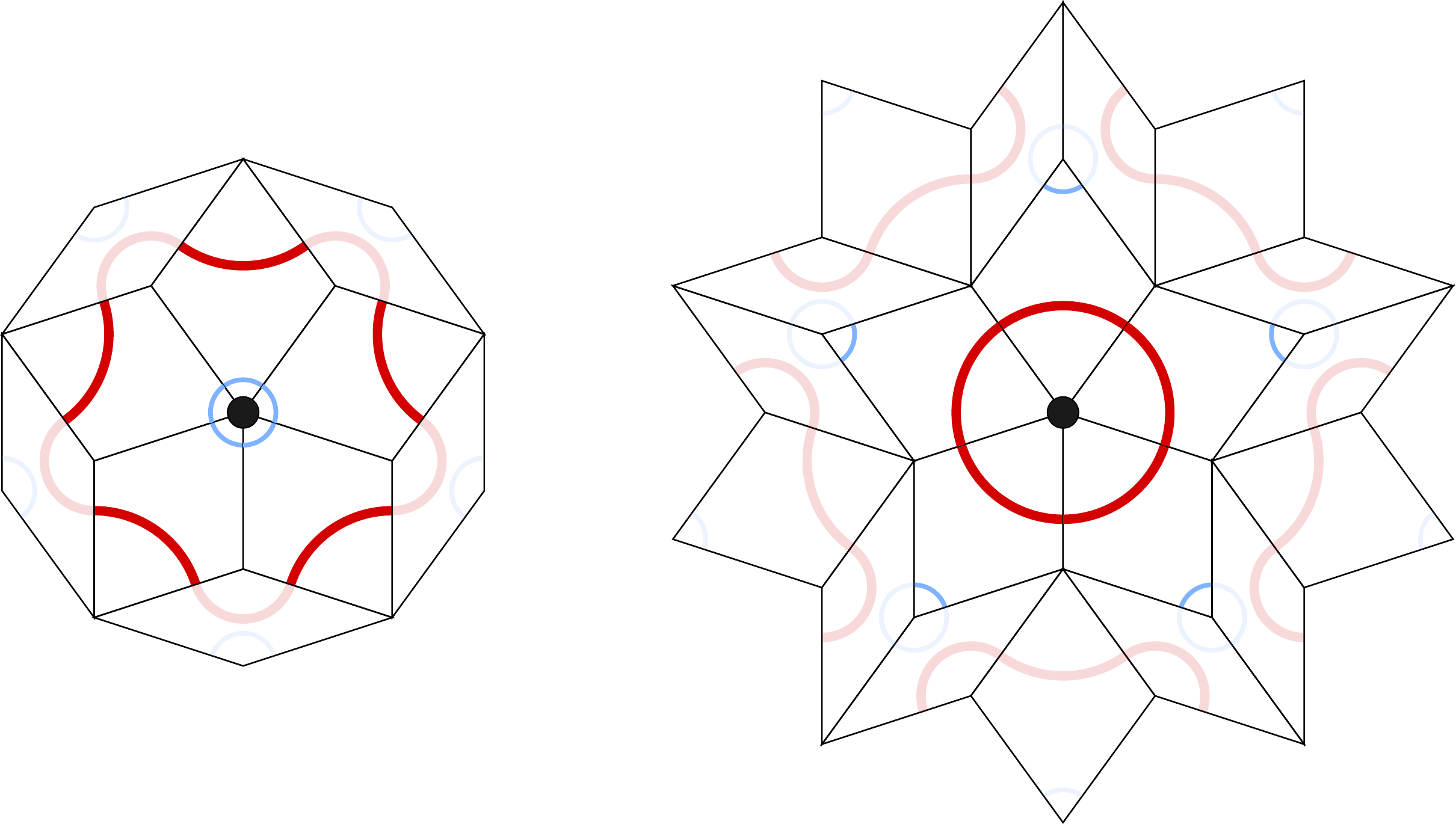}
    \caption{
      When a vertex is the center of a sun/star pattern, the $1$-atlas shows that this $0$-map can be extended in only two ways.
      This yields, once decorating tiles around the neighbour vertices as specified in Fig.~\ref{fig:decorate_tiles}, two possible decorations for this $0$-map.
    }
    \label{fig:decorate_stars}
  \end{figure}
  
  The label on any edge (arrows or coloured arcs) is thus defined by its endpoints.
  The only problem that may occur is that these endpoints yield different arrows!
  To show that this does not happen, we consider the {\em edge atlas} of Penrose tilings, that is, the set of patterns formed by the tiles which contain a vertex of a given edge (Figures~\ref{fig:edge_atlas} and \ref{fig:edge_atlas_decorated}).
  Since each pattern in the edge atlas is included in a $1$-map, the edge atlas can be directly derived from Fig.~\ref{fig:1_atlas}. Otherwise it can also be directly computed through the methods used to prove Prop.~\ref{prop:1_atlas}.
  
  \begin{figure}[htp]
    \includegraphics[width=\textwidth]{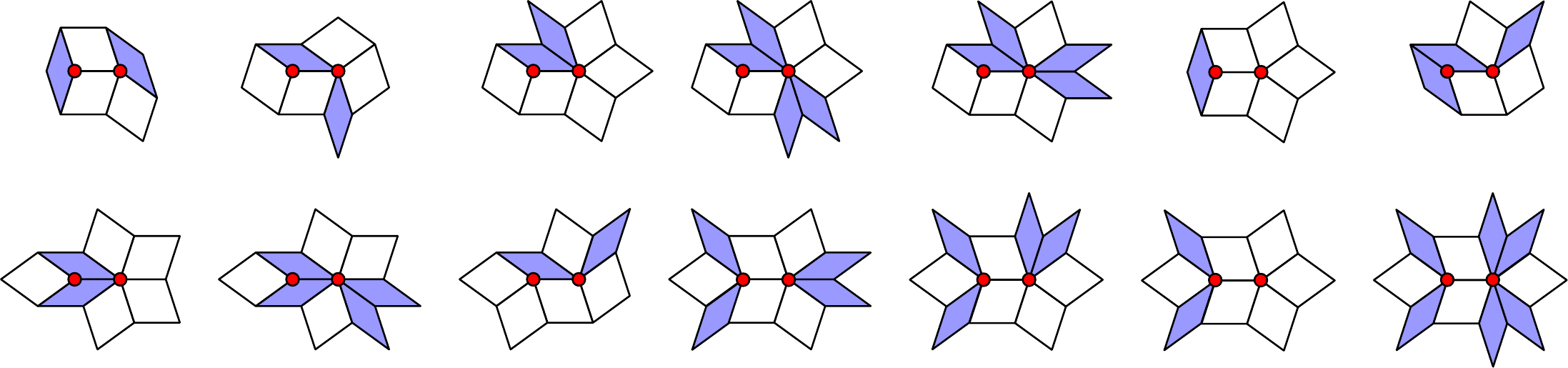}
    \caption{The edge atlas of geometrical Penrose tilings (up to isometry).}
    \label{fig:edge_atlas}
  \end{figure}

  We can now check that the way arrows have been added on tiles around a vertex (Figures~\ref{fig:decorate_tiles} and \ref{fig:decorate_stars}) is consistent for any two neighbour vertices (Figures~\ref{fig:edge_atlas} and \ref{fig:edge_atlas_decorated}). Recall that for the specific case of sun/star pattern, the labelling is consistent accross the pattern and this can be seen either by looking at the edge-atlas patterns around the boundary of the sun/star or by the $1$-atlas which allows only the sun and the star labellings which are both consistent accross the pattern.
  
  Hence, we get a tiling with rhombi labelled as the rhombi of Penrose tilings.
  By definition, this is a Penrose tiling. Hence the original unlabelled tiling is a geometrical Penrose tiling.
  Lemma~\ref{lemma:labelling} is thus proven.
\end{proof}

\begin{figure}[htp]
  \includegraphics[width=\textwidth]{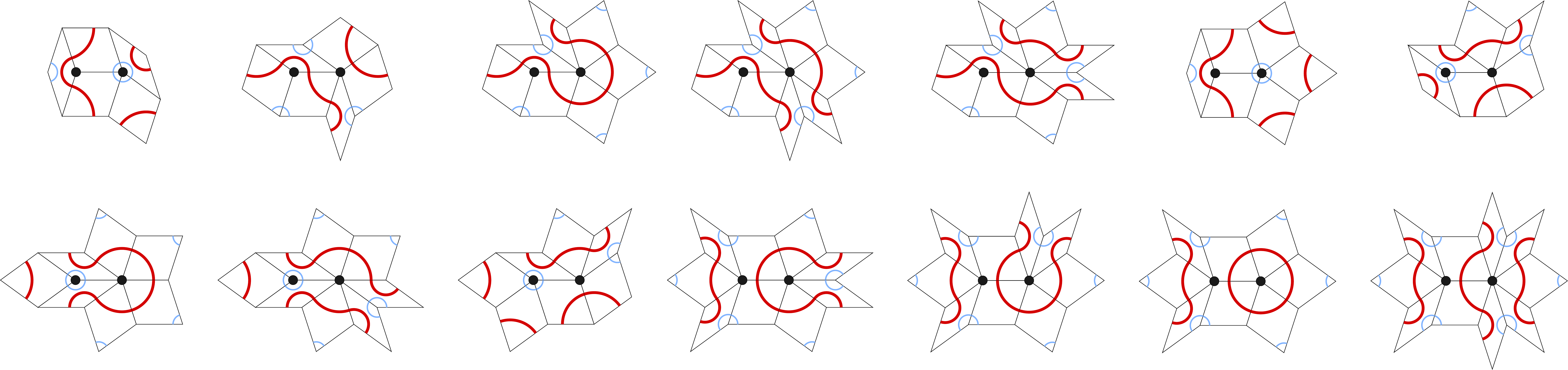}
  \caption{Consistence of the decorations around two connected vertices.}
  \label{fig:edge_atlas_decorated}
\end{figure}

Combining Proposition \ref{prop:1_atlas} and Lemma \ref{lemma:labelling} we obtain Theorem \ref{th:main}.

\begin{remark}
  The edge-atlas of Fig.~\ref{fig:edge_atlas} also characterizes the geometrical Penrose tilings. Indeed the same proof holds for the first 6 $0$-maps, the case of sun/star pattern being slightly more difficult.  
  % Consider a sun/star pattern, such that the patterns along its boundary are all in the edge-atlas presented above.
  % We obtain only the two cases of the $1$-atlas which correspond respectively to a sun pattern and to a star pattern. 
  However, for clarity, we prefer the statement with the $1$-vertex-atlas.
\end{remark}

\printbibliography

\newpage

\appendix

\section{A substitution whose tilings are not linearly recurrent}
\label{appendix:solomyak}
Here we present an example of substitution whose expansion is not a similarity (so Lemma \ref{lemma:solomyak} does not apply) and whose tilings are not linearly recurrent, in essence this example proves the fact that the hypothesis of the expansion being a similarity is necessary in Lemma \ref{lemma:solomyak}.

The substitution $\sigma$ of Figure \ref{fig:solomyak_example} on square tiles with labels $0$ (or $white$) and $1$ (or $black$) due to Solomyak (private communication) is primitive and yields uniformly recurrent tilings that  are not linearly recurrent. Let $\tau$ be the Thue-Morse substitution $\tau: 0 \mapsto 01,\ 1 \mapsto 10$, the substitution $\sigma$ is essentially $\tau$ with $0$-padding above and below.

\begin{figure}[h]
  \center
  \includegraphics[width=0.5\textwidth]{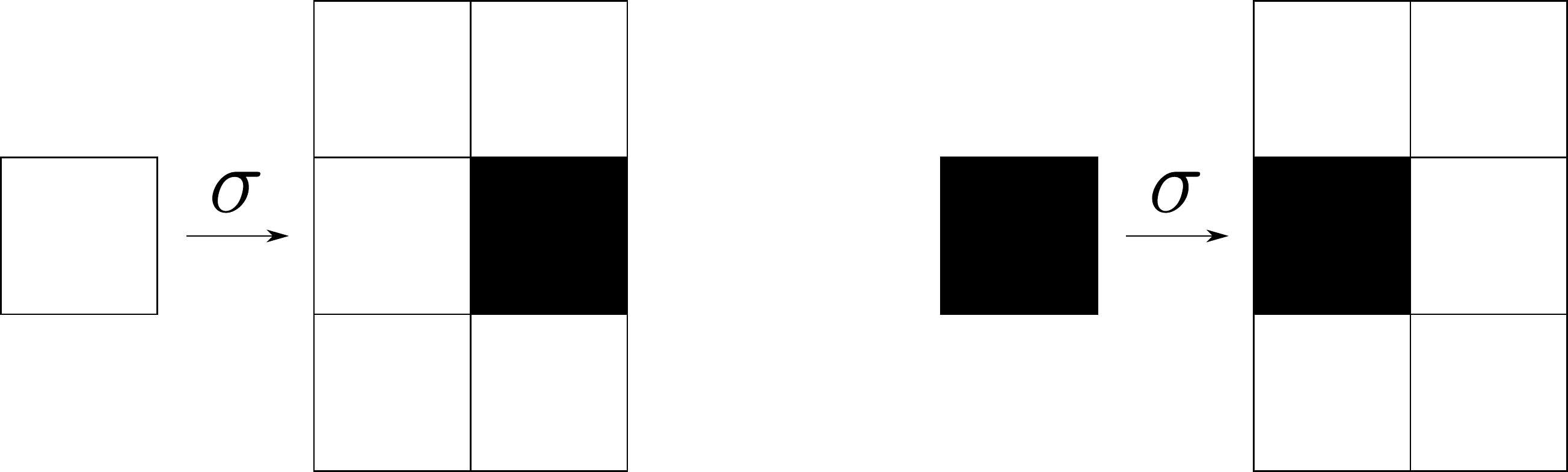}
  \caption{The substitution $\sigma$ on square tiles with labels $0$ (or $white$) and $1$ (or $black$). Note that if we look only at the horizontal line we obtain the Thue Morse substitution, informally the substitution $\sigma$ is just the Thue-Morse substitution $\tau$ with $0$-padding above an below.}
  \label{fig:solomyak_example}
\end{figure}

Recall a known fact on the Thue-Morse substitution : for any $n\geq 0$, the word $\tau^{n+1}(0)$ does not appear as a factor subword in $(\tau^n(0))^\omega$ (the periodic repetition of $\tau^n(0)$).

Let $n$ be a positive integer. The rectangle pattern $\sigma^n(0^{3^n})$: 
\begin{itemize}
\item is legal, indeed for any $k$ the pattern $0^k$ is legal because it appears in the topmost line of $\sigma^{\lceil \log_2(k) \rceil}(0)$ so for any $k$ and any $n$, $\sigma^n(0^k)$ is legal because it appears in $\sigma^{n+\lceil \log_2(k) \rceil}(0)$ and in particular it holds for $k=3^n$,
  
\item has size $6^n\times 3^n$  indeed $\sigma^n(0^{3^n}) = (\sigma^n(0))^{3^n}$ and each $\sigma^n(0)$ has size $2^n\times 3^n$. In particular $\sigma^n(0^{3^n})$ has size more than $3^n\times 3^n$,
  
\item  does not contain the word $\tau^{n+1}(0)$ which is of size $2^{n+1} \times 1$ and which appears in $\sigma^{n+1}(0)$.
\end{itemize}

Let $\tiling$ be a $\sigma$-tiling, by definition $\tiling$ contains $\sigma^n(0)$ for any $n$. In particular, for any $n$ it contains $\sigma^n(0^{3^n})$ and $\tau^{n+1}(0)$.
So for any $n$, there exists a pattern of diameter $2^{n+1}$ that appears in the tiling, but does not appear in regions of size $3^n\times 3^n$, this is a contradiction to linear recurrence as the ratio $3^n/2^{n+1}$ is unbounded.

\end{document}